\documentclass[lettersize,journal]{IEEEtran}

\usepackage{amsmath,amssymb,amsfonts,amsthm}
\usepackage{bm}
\usepackage{graphicx}
\usepackage{array}
\usepackage{textcomp}
\usepackage{stfloats}
\usepackage{url}
\usepackage{cite}
\usepackage{mathrsfs}
\usepackage{float}
\usepackage{color}
\usepackage{nicematrix}
\usepackage{booktabs}
\usepackage{algorithm}
\usepackage{algorithmicx}
\usepackage{algpseudocodex}
\usepackage{multirow}
\usepackage{enumerate}
\usepackage{mathtools}
\usepackage[normalem]{ulem}

\allowdisplaybreaks[4]
\newcommand{\perturbed}[1]{{#1}^{\Delta}}

\newcommand{\p}{\mathrm{p}}
\renewcommand{\c}{\mathrm{c}}
\renewcommand{\i}{\mathrm{i}}
\newcommand{\tail}{\mathrm{tail}}

\newcommand{\Frobenius}[1]{\left\|#1\right\|_{\mathrm{F}}}

\newcommand{\define}{\coloneqq}
\newcommand{\set}[1]{\left\{#1\right\}}

\renewcommand{\ker}{\operatorname{ker}}

\newcommand{\twonorm}[1]{\left\|#1\right\|_2}

\newcommand{\abs}[1]{\left|#1\right|}
\newcommand{\krylov}{\mathcal{K}}
\newcommand{\spanspace}[1]{\operatorname{colspan}\left\{#1\right\}}

\newcommand{\tol}{\mathrm{tol}}
\newcommand{\imaginary}{\imath}
\newcommand{\real}{\mathbb{R}}
\newcommand{\complex}{\mathbb{C}}

\renewcommand{\Re}{\operatorname{Re}}

\newcommand{\diag}[1]{\operatorname{diag}\left\{#1\right\}}

\newcommand{\tran}[1]{#1^\top}
\newcommand{\inv}[1]{#1^{-1}}

\newcommand{\invtran}[1]{#1^{-\top}}

\newcommand{\conjtran}[1]{#1^\ast}
\newcommand{\invconjtran}[1]{#1^{-\ast}}

\newcommand{\secref}[1]{Section \ref{#1}}
\newcommand{\thmref}[1]{Theorem \ref{#1}}

\newcommand{\propref}[1]{Proposition \ref{#1}}

\newcommand{\tableref}[1]{Table \ref{#1}}
\newcommand{\figref}[1]{Figure \ref{#1}}
\newcommand{\appref}[1]{Appendix \ref{#1}}

\makeatletter
\newif\ifrevdel@complex
\protected\long\def\revdel#1{%
    \begingroup
    \revdel@complexfalse
    \edef\revdel@text{\detokenize{#1}}%
    \edef\revdel@begin{\detokenize{\begin}}%
    \edef\revdel@check{\noexpand\in@{\revdel@begin}{\revdel@text}}%
    \revdel@check
    \ifin@\revdel@complextrue\fi
    \edef\revdel@cite{\detokenize{\cite}}%
    \edef\revdel@check{\noexpand\in@{\revdel@cite}{\revdel@text}}%
    \revdel@check
    \ifin@\revdel@complextrue\fi
    \edef\revdel@texttt{\detokenize{\texttt}}%
    \edef\revdel@check{\noexpand\in@{\revdel@texttt}{\revdel@text}}%
    \revdel@check
    \ifin@\revdel@complextrue\fi
    \edef\revdel@dollar{\detokenize{$}}%
    \edef\revdel@check{\noexpand\in@{\revdel@dollar}{\revdel@text}}%
    \revdel@check
    \ifin@\revdel@complextrue\fi
    \ifrevdel@complex
        {\color{red}#1}%
    \else
        \textcolor{red}{\sout{#1}}%
    \fi
    \endgroup
}
\makeatother

\theoremstyle{plain}
\newtheorem{theorem}{Theorem}[section]

\newtheorem{proposition}[theorem]{Proposition}

\theoremstyle{definition}

\theoremstyle{remark}
\newtheorem{remarkx}{Remark}[section]

\floatname{algorithm}{Algorithm}

\renewcommand{\algref}[1]{Algorithm \ref{#1}}

\begin{document}

\title{FlexRC: A Flexible Multi-Point Model Order Reduction Method for Many-Port RC Networks}
\author{Yuncheng Xu, Siyuan Yin, Lin Liu, Fan Yang,~\IEEEmembership{Member,~IEEE}, Xuan Zeng,~\IEEEmembership{Senior Member,~IEEE}, Chengtao An, and Yangfeng Su%
\thanks{Yuncheng Xu, Siyuan Yin, and Yangfeng Su are with the School of Mathematical Sciences, Fudan University, Shanghai, China.}
\thanks{Lin Liu and Chengtao An are with Empyrean, China.}
\thanks{Fan Yang and Xuan Zeng are with the State Key Laboratory of Integrated Chips and Systems, College of Integrated Circuits and Micro-Nano Electronics, Fudan University, Shanghai, China.}
\thanks{Corresponding authors: Chengtao An (e-mail: ancht@empyrean.com.cn) and Yangfeng Su (e-mail: yfsu@fudan.edu.cn).}
}

% \markboth{Manuscript Submitted to IEEE Transactions on Computer-Aided Design of Integrated Circuits and Systems}%


\maketitle

\begin{abstract}
    Efficient model order reduction for many-port resistor-capacitor (RC) networks is essential in post-layout circuit simulation. Existing high-accuracy elimination-based methods have certain limitations, such as fixed frequency points, large reduced-order models, or high reduction cost. This paper proposes FlexRC, a flexible multi-point model order reduction method for many-port RC networks. FlexRC starts from the same elimination step as previous methods, and then constructs a nonorthogonal projection basis by a modified block rational Arnoldi process to generate a sparse banded reduced model. FlexRC features three adjustable components: user-specified frequency points, a tolerance-controlled port-reduction technique for the internal subsystem, and an optional sparsity-control strategy. We discuss passivity under port-reduction perturbations, analyze moment matching, and provide a conservative error estimate for port reduction. Numerical experiments on industrial RC examples and IBM power-grid examples demonstrate the effectiveness of FlexRC in terms of reduction time and transient simulation time. 
\end{abstract}


\section{Introduction}
In modern IC design, interconnect effects have become a dominant factor in determining whole-chip performance \cite{Chen1998interconnect, Silva2007issues, Nassif2008power}. Massive parasitic models, especially resistor-capacitor (RC) networks, are extracted and connected to nonlinear devices for post-layout simulation. 
The large numbers of nodes and ports in RC parasitic networks make direct nonlinear simulation computationally prohibitive and time-consuming. 

Model order reduction (MOR) is frequently used to speed up the simulation of interconnect circuits. MOR methods applied directly to the full nonlinear system, such as proper orthogonal decomposition (POD) \cite{AK2004POD, RP2003POD, Pinnau2008POD}, are computationally expensive. Therefore, RC reduction (RCR), which reduces the linear time-invariant RC networks before nonlinear simulation, has become the mainstream approach. An ideal RC reduction technique should yield accurate reduced-order models while preserving essential properties such as passivity \cite{OCP1998PRIMA,OCP1999PRIMA} and input-output structure \cite{YZSZ2007RLCSYN}. The primary challenge in RCR lies in the large number of ports connecting the RC network to nonlinear devices. The input-output structure associated with these ports must be preserved during the reduction process to ensure that the reduced-order model can still be connected to the nonlinear devices for further analysis. 

High-accuracy MOR methods are difficult to apply to many-port RC networks. TBR-like methods \cite{RS2010PABTEC, RS2011lyapunov, PS2004PMTBR} are limited by slowly decaying Hankel singular values and the high cost of Lyapunov equations \cite{SP2004exploiting}, while factor division algorithms \cite{PP2018model, PP2019model} become impractical when both node and port counts are large. 

Traditional Krylov subspace methods, notably PRIMA \cite{OCP1998PRIMA,OCP1999PRIMA}, are widely used but become inefficient for many-port networks because the dense projection matrix often produces dense reduced models \cite{silva2007outstanding}. Port-compression methods such as SVDMOR \cite{Feldmann2004SVDMOR}, ESVDMOR \cite{LTYM2008ESVDMOR}, and RECMOR \cite{FL2004RECMOR} rely on port correlations that are often weak in practical networks \cite{YTZCS2012decentralized}. Furthermore, these methods do not preserve the input-output structure. Alternative approaches such as SPRIM \cite{Freund2004SPRIM} and RLCSYN \cite{YZSZ2007RLCSYN} preserve this structure, but suffer from increased model orders. 

State-of-the-art RC reduction methods primarily rely on elimination-based methods such as PACT \cite{KY1996PACT}, TICER \cite{Sheehan1999TICER}, and SIP \cite{YDZP2008SIP}. Instead of constructing an explicit projection matrix, these methods use sparse Gaussian elimination to generate reduced-order models that naturally preserve the original input-output structure. Enhanced by the partition strategy in SparseRC \cite{IRS2011SparseRC}, these techniques often work well. However, they typically match only two DC moments, which may be insufficient for high-accuracy applications \cite{YDZP2008SIP}. While multi-point matching was introduced in \cite{YDZP2008SIP}, it may result in singular reduced matrices. TurboMOR-RC \cite{OT2016TurboMOR} extends SIP to higher-order moment matching and generates banded sparse reduced-order models that can be solved efficiently, but its frequency point is fixed at $s_0 = 0$. SMP-RCR \cite{YXLYZAS2025SMPRCR} proposes a sparse multi-point moment-matching method, but it requires the explicit generation of large dense intermediate blocks, which impairs its practicality for large-scale cases. 

This paper proposes FlexRC, a flexible multi-point model order reduction method for many-port RC networks. Starting from the standard port-preserving elimination step, FlexRC applies a port-reduction technique to the coupling block between port nodes and internal nodes, thereby reducing the port number of the internal subsystem. A block rational Arnoldi process with incomplete orthogonalization is then used to construct a nonorthogonal basis for the projection space. This construction allows FlexRC to generate sparse banded reduced-order models similar to those of TurboMOR-RC. FlexRC is also compatible with the sparsity-control strategy used in SMP-RCR. Thus FlexRC combines user-specified frequency points with a TurboMOR-RC-like sparse block structure, without explicitly constructing the dense intermediate blocks required in SMP-RCR. 

The flexibility of FlexRC comes from three adjustable components. First, the frequency points can be specified by the user to improve accuracy for the target transient response. Second, the port-reduction tolerance $\tol$ controls the error introduced by the port reduction and provides a practical trade-off between accuracy and reduced order. With suitable choices, FlexRC can significantly reduce the order with only a small loss of accuracy. Even in the conservative case $\tol = 0$, the resulting reduced models have simulation efficiency comparable to that of TurboMOR-RC. Third, the sparsity-control variant retains more nodes during elimination to generate a sparser reduced model and avoid overly dense reduced matrices. These mechanisms are especially useful for reducing the subsequent simulation time. 

In addition to the proposed reduction framework, we discuss passivity under port-reduction perturbations. For the exact reduction, passivity follows from the congruence structure of the construction. When port reduction is applied, the reduced capacitance matrix is symmetrically perturbed, and passivity is preserved if this perturbation is sufficiently small. We also analyze moment matching and provide a conservative error estimate for port reduction. Beyond the standard moment-matching analysis, we also consider RC networks with singular conductance matrices and prove a Laurent moment matching result for elimination-based RC reduction methods. To the best of our knowledge, this is the first treatment of this case for elimination-based RC reduction methods. Numerical experiments on industrial RC examples and IBM power-grid examples demonstrate the effectiveness of FlexRC in terms of reduction time and simulation time. The results also show the benefits of the three adjustable components in our flexible framework. 

The rest of this paper is organized as follows. In \secref{sec:preliminary}, we describe the RC reduction problem and review projection-based and elimination-based methods. In \secref{sec:proposed method}, we present the FlexRC method and its implementation details. In \secref{sec:analysis}, we analyze the passivity, moment-matching properties, and port-reduction error bound. Numerical experiments are reported in \secref{sec:numerical experiments}. Finally, conclusions are drawn in \secref{sec:conclusion}, and some mathematical proofs are provided in the Appendix. 

\section{RC reduction and previous methods}
\label{sec:preliminary}

This section introduces the basic notation for RC reduction and the concept of moment matching, followed by a review of projection-based Krylov subspace methods and elimination-based methods for RC reduction. 

\subsection{RC reduction and moment matching}

Using modified nodal analysis \cite{HRB1975MNA}, an RC network with $N$ nodes and $p$ ports can be described by the following differential-algebraic equations 
\begin{align}
    \label{eq:circuit equations in time domain}
    \left\{\begin{aligned}
        & C\dot{x}(t) + Gx(t) = Bu(t), \\
        & y(t) = \tran{B}x(t), 
    \end{aligned}\right.
\end{align}
where the input vector $u$ collects all port currents, the output vector $y$ collects all port voltages, and the vector $x$ contains all nodal voltages. Matrices $G, C \in \real^{N \times N}$ represent the conductance and capacitance matrices, respectively, and both are symmetric positive semidefinite. The incidence matrix $B \in \real^{N \times p}$ relates the input sources to the corresponding nodes. In this paper, we assume that the nodes have been permuted such that the port nodes are ordered first and the internal nodes follow. In this ordering, $B$ can be written as 
\begin{align*}
    B = \begin{bmatrix}
        I_p \\ 0_{(N-p) \times p} 
    \end{bmatrix}. 
\end{align*}
Thus the first $p$ nodal variables correspond exactly to the port nodes. 

The objective of RC reduction (RCR) is to approximate \eqref{eq:circuit equations in time domain} with a reduced-order model of order $r \ll N$, 
\begin{align}
    \label{eq:reduced circuit equations in time domain}
    \left\{\begin{aligned}
        & \tilde{C}\dot{\tilde{x}}(t) + \tilde{G}\tilde{x}(t) = \tilde{B}u(t), \\
        & \tilde{y}(t) = \tran{\tilde{B}}\tilde{x}(t), 
    \end{aligned}\right.
\end{align}
where $\tilde{G}, \tilde{C} \in \real^{r \times r}$ and $\tilde{B} \in \real^{r \times p}$. In addition to accuracy, the reduced matrices $\tilde{G}$ and $\tilde{C}$ must remain symmetric positive semidefinite to ensure symmetry, stability, and passivity \cite{YDZP2008SIP,IRS2011SparseRC,RS2011lyapunov}. Furthermore, the primary challenge in RCR stems from the large number of ports connecting the parasitic RC network to nonlinear devices. The input-output structure must be preserved so that the reduced model can be synthesized into a compatible equivalent RC netlist for further analysis \cite{YZSZ2007RLCSYN,IRS2011SparseRC}. Consequently, the reduced incidence matrix $\tilde{B}$ should retain the form 
\begin{align*}
    \tilde{B} = \begin{bmatrix}
        I_p \\ 0_{(r-p) \times p} 
    \end{bmatrix}. 
\end{align*}

A standard approach to ensuring approximation accuracy is through moment matching. The transfer functions of the original model \eqref{eq:circuit equations in time domain} and the reduced-order model \eqref{eq:reduced circuit equations in time domain} are given by $H(s) = \tran{B}\inv{(G + s C)}B$ and $\tilde{H}(s) = \tran{\tilde{B}}\inv{(\tilde{G} + s \tilde{C})}\tilde{B}$, respectively. Given an expansion point $s_0 \geq 0$, let $A(s_0) \define G + s_0 C$ and assume that $A(s_0)$ is nonsingular. Then the Taylor expansion of $H(s)$ around $s = s_0$ can be written as 
\begin{align*}
    H(s) & = \sum_{k=0}^{\infty} M_k(s_0) (s-s_0)^k,
\end{align*}
where the $k$-th order moments $M_k(s_0)$ are given by 
\begin{align*}
    M_k(s_0) & \define (-1)^k \tran{B} (\inv{A(s_0)}C)^k \inv{A(s_0)}B. 
\end{align*}
Similarly, $\tilde{M}_k(s_0)$ can be defined for the reduced-order model. The reduced system is said to match the first $m$ moments around $s = s_0$ if 
\begin{align*}
    \tilde{M}_k(s_0) = M_k(s_0), \qquad k = 0, 1, \ldots, m-1. 
\end{align*}
For RC reduction, a particularly important expansion point is $s_0 = 0$, which corresponds to the DC moments \cite{OCP1998PRIMA}. 

\subsection{Projection-based methods}

Projection-based Krylov subspace methods provide a standard way to achieve moment matching. The order-$q$ block Krylov subspace is defined as $\krylov_q(A, B) \define \spanspace{B, AB, \ldots, A^{q-1}B}$. PRIMA \cite{OCP1998PRIMA, OCP1999PRIMA} is a classical and widely used Krylov-subspace MOR method in circuit simulation. In PRIMA, a block Arnoldi process is performed to generate an orthogonal matrix $V$ whose columns span the Krylov subspace $\krylov_q(\inv{A(s_0)}C, \inv{A(s_0)}B)$. The reduced-order model is subsequently obtained by the congruence transformations 
\begin{align*}
    \tilde{C} = \tran{V}CV, \quad \tilde{G} = \tran{V}GV, \quad \tilde{B} = \tran{V}B. 
\end{align*}
The following standard result from \cite{LS2012krylov} shows that the resulting reduced model matches the first $2q$ moments of the original system around $s = s_0$. 
\begin{theorem}
    \label{thm:moment-matching for RC reduction}
    Suppose $G + s_0 C$ and $\tilde{G} + s_0 \tilde{C}$ are both nonsingular. If the subspace $\spanspace{V}$ contains the Krylov subspace $\krylov_{q}(\inv{(G + s_0 C)}C, \inv{(G + s_0 C)}B)$, then the reduced-order system matches the first $2q$ moments around $s = s_0$. 
\end{theorem}

The reduced-order model generated by PRIMA is passive, but it does not maintain the specific input-output structure required for RC networks. To address this limitation, methods such as SPRIM \cite{Freund2004SPRIM} and RLCSYN \cite{YZSZ2007RLCSYN} partition the matrix $V$ into two blocks denoted by $V_{\p}$ and $V_{\i}$ 
\begin{align*}
    V = \begin{bmatrix}
        V_{\p} \\ V_{\i} 
    \end{bmatrix}, \qquad 
    V_{\p} \in \real^{p \times r}, \quad V_{\i} \in \real^{(N-p) \times r}. 
\end{align*}
The projection matrix is then constructed as 
\begin{align*}
    V_{\mathrm{R}} = \begin{bmatrix}
        I_p & \\
        & \mathrm{orth}(V_{\i}) 
    \end{bmatrix},
\end{align*}
where $\mathrm{orth}(V_{\i})$ denotes the orthogonalization of the submatrix $V_{\i}$. The moment-matching property remains satisfied by the reduced-order model, but the reduced order increases by $p$, and the orthogonalization of $V_{\i}$ can be expensive. 

\subsection{Elimination-based methods}

The reduced matrices $\tilde{G}$ and $\tilde{C}$ generated by projection-based subspace methods are typically dense, so solving the reduced-order model can still be computationally expensive. To address this issue, elimination-based methods \cite{YDZP2008SIP,IRS2011SparseRC,OT2016TurboMOR} have been proposed. These methods naturally preserve the port nodes and the original input-output structure during elimination. 

SIP \cite{YDZP2008SIP} performs projections implicitly through Gaussian elimination. The circuit equations \eqref{eq:circuit equations in time domain} are partitioned into port nodes and internal nodes as follows 
\begin{align}
    \label{eq:blocked circuit equations in frequency domain}
    \left(\begin{bmatrix}
        G_{\p} & \tran{G_{\c}} \\
        G_{\c} & G_{\i} 
    \end{bmatrix} + s 
    \begin{bmatrix}
        C_{\p} & \tran{C_{\c}} \\
        C_{\c} & C_{\i} 
    \end{bmatrix}\right)
    \begin{bmatrix}
        x_{\p} \\ x_{\i} 
    \end{bmatrix} = 
    \begin{bmatrix}
        I_p \\ 0 
    \end{bmatrix} u, 
\end{align}
where $G_{\p}, C_{\p} \in \real^{p \times p}$ represent the contributions between port nodes; $G_{\i}, C_{\i} \in \real^{(N-p) \times (N-p)}$ represent the contributions between internal nodes; and $G_{\c}, C_{\c} \in \real^{(N-p) \times p}$ describe the coupling between internal and port nodes. 

For the DC expansion point $s_0 = 0$, assume that $G_{\i}$ is nonsingular and define $V_1 = \inv{G_{\i}}G_{\c}$. SIP constructs the transformation matrix 
\begin{align*}
    M = \begin{bmatrix}
        I_p & 0 \\
        -V_1 & I_{N-p} 
    \end{bmatrix}, 
\end{align*}
and applies congruence transformations with $M$ to the original system \eqref{eq:circuit equations in time domain}, yielding 
\begin{align}
    \label{eq:SIP congruence transformation}
    \tran{M}GM & = \begin{bmatrix}
        \tilde{G}_{\p} & 0 \\ 
        0 & G_{\i} 
    \end{bmatrix}, \quad 
    \tran{M}CM & = \begin{bmatrix}
        \tilde{C}_{\p} & \tran{\tilde{C}_{\c}} \\ 
        \tilde{C}_{\c} & C_{\i} 
    \end{bmatrix}. 
\end{align}
This congruence transformation eliminates the coupling block $G_{\c}$ in the conductance matrix and can be implemented by Gaussian elimination. During this process, the internal submatrices $G_{\i}, C_{\i}$ and the incidence matrix $B$ remain invariant. The reduced-order model is then obtained by eliminating the internal nodes: 
\begin{align*}
    \tilde{G} = \tilde{G}_{\p}, \quad \tilde{C} = \tilde{C}_{\p}, \quad \tilde{B} = I_p. 
\end{align*}
The reduced-order system generated by SIP is passive and matches the first two moments around $s = 0$. 

TurboMOR-RC \cite{OT2016TurboMOR} extends SIP to match $2q$ moments at $s = 0$. Starting with \eqref{eq:SIP congruence transformation}, TurboMOR-RC first normalizes $G_{\i}$ to $I_{N-p}$ using the Cholesky factor $L$ of $G_{\i}$. The coupling matrix is then decomposed by Householder QR as 
\begin{align*}
    \inv{L}\tilde{C}_{\c} = Q^{(2)}
    \begin{bmatrix}
        R^{(2)} \\ 0
    \end{bmatrix}, 
\end{align*}
where $Q^{(2)} \in \real^{(N-p) \times (N-p)}$ is the product of a series of Householder reflections, and $R^{(2)} \in \real^{p\times p}$ is upper triangular. After the corresponding congruence transformation, the original system can be partitioned as 
\begin{align*}
    G^{(2)} & = \begin{bmatrix}
        \tilde{G}_{\p} & & \\
        & I_p & \\
        & & I_{N-2p} 
    \end{bmatrix}, \\
    C^{(2)} & = \begin{bmatrix}
        \tilde{C}_{\p} & \tran{\left(R^{(2)}\right)} & 0 \\
        R^{(2)} & C_{\p}^{(2)} & \tran{\left(C_{\c}^{(2)}\right)} \\
        0 & C_{\c}^{(2)} & C_{\i}^{(2)} 
    \end{bmatrix}. 
\end{align*}
Repeating this procedure for another $q-2$ stages and retaining the first $qp$ variables yields the reduced-order model of TurboMOR-RC. The resulting reduced matrices have a banded structure, with both lower and upper bandwidths equal to $p$, so the reduced-order model can be solved more efficiently. 

In practical implementation, the matrices that are not retained in the final reduced-order model need not be explicitly generated. For example, for $q = 2$, the discarded blocks $C_{\c}^{(2)}$ and $C_{\i}^{(2)}$ are not needed and therefore need not be generated. This is important because these discarded blocks can be very large and dense, so explicitly forming them can require substantial computation time and memory. 

SMP-RCR \cite{YXLYZAS2025SMPRCR} is a multi-point moment-matching method. Unlike TurboMOR-RC, whose frequency point is fixed at $s_0 = 0$, SMP-RCR allows the frequency sequence to be chosen freely to improve accuracy. It also applies a deflation strategy and sparsity-control technique to improve the efficiency of the reduced-order model. However, SMP-RCR must explicitly generate the intermediate discarded blocks that are not retained in the final reduced-order model. Since these blocks can be very large and dense, the reduction time and memory cost can become prohibitive for large-scale cases. 

\section{Proposed method}
\label{sec:proposed method}

In this section, we present FlexRC, a multi-point moment-matching reduction method for RC networks. Let 
\begin{align*}
    \mathcal{S} \define [s_1, s_2, \ldots, s_q]
\end{align*}
denote the prescribed sequence of frequency points, where $s_1, s_2, \ldots, s_q$ are nonnegative real numbers. For RC reduction, the first frequency point is chosen to be $s_1 = 0$ to match the DC moments. To make the construction clear, we first describe the two-point case and then extend it to $q \geq 3$. The corresponding moment-matching property and error estimates are analyzed in the next section. 

\subsection{Two-point case}
We first describe the reduction procedure for the two-point case with frequency sequence $[s_1 = 0, s_2]$. Starting from \eqref{eq:SIP congruence transformation}, the conductance coupling between the port variables and the internal variables has been eliminated, while the remaining interaction is carried by the capacitive block $\tilde{C}_{\c}$. Similar to TurboMOR-RC, FlexRC keeps the effect of the internal subsystem and reduces the following internal subsystem at the additional frequency point: 
\begin{align}
    \label{eq:internal subsystem}
    \Sigma_{\i}: \left\{\begin{aligned} 
        & C_{\i} \dot{x}_{\i} + G_{\i} x_{\i} = \tilde{C}_{\c}u_{\i}, \\ 
        & y_{\i} = \tran{\tilde{C}_{\c}}x_{\i}. 
    \end{aligned}\right. 
\end{align}

Before reduction, we first perform port reduction on the internal coupling matrix $\tilde{C}_{\c}$. Instead of performing a full QR decomposition in TurboMOR-RC or SMP-RCR, we compute an economic QR decomposition 
\begin{align*}
    \tilde{C}_{\c} = QR, 
\end{align*}
where $Q \in \real^{(N-p) \times p}$ has orthonormal columns and $R \in \real^{p \times p}$ is upper triangular. 
The subsequent port reduction is performed by row norms. Let 
\begin{align*}
    \rho_j \define \twonorm{R(j, :)}, \qquad j = 1, \ldots, p, 
\end{align*}
and let $\pi$ be a permutation such that $\rho_{\pi_1} \leq \cdots \leq \rho_{\pi_p}$. Given a prescribed port-reduction tolerance $\tol$, we choose the largest integer $m \in \set{0,1,\ldots,p}$ such that 
\begin{align*}
    \frac{\sqrt{\rho_{\pi_1}^2 + \cdots + \rho_{\pi_m}^2}}{\Frobenius{R}} \leq \tol. 
\end{align*}
The rows indexed by 
\begin{align*}
    \mathcal{J} \define \set{\pi_1,\ldots,\pi_m}
\end{align*}
are discarded, and the remaining rows 
\begin{align*}
    \mathcal{I} \define \set{1,\ldots,p}\setminus\mathcal{J}
\end{align*}
are retained. Set $p_{\i} = \abs{\mathcal{I}}$ and define 
\begin{align*}
    B_{\i} = Q(:,\mathcal{I}), \qquad R_{\i} = R(\mathcal{I}, :). 
\end{align*}
After discarding these rows, we obtain the low-rank approximation 
\begin{align*}
    \tilde{C}_{\c} \approx B_{\i}R_{\i}. 
\end{align*}
Since $Q$ has orthonormal columns, the truncation error satisfies 
\begin{align*}
    \frac{\Frobenius{\tilde{C}_{\c} - B_{\i}R_{\i}}}{\Frobenius{\tilde{C}_{\c}}}
    = \frac{\Frobenius{Q(:,\mathcal{J})R(\mathcal{J}, :)}}{\Frobenius{QR}}
    = \frac{\Frobenius{R(\mathcal{J}, :)}}{\Frobenius{R}}
    \leq \tol. 
\end{align*}
Define 
\begin{align*}
    \Delta \define B_{\i}R_{\i} - \tilde{C}_{\c}. 
\end{align*}
Thus the port reduction step can be interpreted as a relative backward perturbation of the internal coupling matrix: 
\begin{align*}
    \Frobenius{\Delta} / \Frobenius{\tilde{C}_{\c}} \leq \tol. 
\end{align*}
The corresponding perturbed internal subsystem is given by 
\begin{align}
    \label{eq:perturbed internal subsystem}
    \perturbed{\Sigma}_{\i}: \left\{\begin{aligned} 
        & (G_{\i} + s C_{\i})x_{\i} = (B_{\i}R_{\i})u_{\i}, \\ 
        & y_{\i} = \tran{(B_{\i}R_{\i})}x_{\i}. 
    \end{aligned}\right. 
\end{align}
The following reduction is applied to this perturbed subsystem. Since only the $p_{\i}$ columns of $B_{\i}$ are used to generate the internal basis, the dimension of the reduced-order model is reduced accordingly. 

For the second frequency point $s_2$, we apply a Krylov-subspace projection to the perturbed internal subsystem. To preserve the sparsity inherited from $R_{\i}$, however, we do not construct an orthonormal block basis. First, we compute the $0$-th moment block by 
\begin{align*}
    \hat{V}_2 D = \inv{(G_{\i} + s_2 C_{\i})}B_{\i}, 
\end{align*}
where $D$ is a nonsingular diagonal matrix induced by column normalization. In practice, $\hat{V}_2$ is obtained by normalizing each column of $\inv{(G_{\i} + s_2 C_{\i})}B_{\i}$ separately. This step is only a column scaling to improve numerical stability, so $\hat{V}_2$ is not an orthonormal basis in general. Let 
\begin{align*}
    T_2 = \tran{B_{\i}} \hat{V}_2. 
\end{align*}
Note that 
\begin{align*}
    T_2 = \tran{B_{\i}}\inv{(G_{\i} + s_2 C_{\i})}B_{\i}\inv{D}. 
\end{align*}
The matrix $G_{\i} + s_2 C_{\i}$ is symmetric positive definite, $B_{\i}$ has full column rank, and $D$ is nonsingular. Therefore, the matrix $T_2$ is nonsingular. We then apply the small change of basis 
\begin{align*}
    V_2 = \hat{V}_2\inv{T_2}. 
\end{align*}
After this basis change, it holds that 
\begin{align*}
    \tran{V_2}B_{\i} = I_{p_{\i}}. 
\end{align*}
Projecting the perturbed subsystem \eqref{eq:perturbed internal subsystem} with $V_2$ gives 
\begin{align*}
    \tilde{G}_{\i} = \tran{V_2}G_{\i}V_2, \qquad 
    \tilde{C}_{\i} = \tran{V_2}C_{\i}V_2, \qquad 
    \tilde{B}_{\i} = \tran{V_2}B_{\i}R_{\i}. 
\end{align*}
With $\tran{V_2}B_{\i} = I_{p_{\i}}$, the reduced coupling matrix is exactly 
\begin{align*}
    \tilde{B}_{\i} = R_{\i}. 
\end{align*}
Thus the sparsity obtained from the truncated coefficient matrix $R_{\i}$ is inherited directly by the reduced coupling block. Since $s_1 = 0$, the final reduced-order model is 
\begin{align*}
    \tilde{G} = 
    \begin{bmatrix}
        \tilde{G}_{\p} & 0 \\ 
        0 & \tilde{G}_{\i}
    \end{bmatrix}, \qquad 
    \tilde{C} = 
    \begin{bmatrix}
        \tilde{C}_{\p} & \tran{R_{\i}} \\ 
        R_{\i} & \tilde{C}_{\i} 
    \end{bmatrix}, \qquad 
    \tilde{B} = 
    \begin{bmatrix}
        I_p \\ 0
    \end{bmatrix}. 
\end{align*}

In practical implementation, we do not explicitly form $V_2 = \hat{V}_2\inv{T_2}$. Instead, we first project the internal matrices using $\hat{V}_2$, and then apply the small transformation induced by $T_2$ to the reduced matrices. More precisely, if 
\begin{align*}
    \hat{G}_{\i} = \tran{\hat{V}_2}G_{\i}\hat{V}_2, \qquad 
    \hat{C}_{\i} = \tran{\hat{V}_2}C_{\i}\hat{V}_2, 
\end{align*}
then 
\begin{align*}
    \tilde{G}_{\i} = \invtran{T_2}\hat{G}_{\i}\inv{T_2}, \qquad 
    \tilde{C}_{\i} = \invtran{T_2}\hat{C}_{\i}\inv{T_2}, \qquad 
    \tilde{B}_{\i} = R_{\i}. 
\end{align*}
The resulting computational procedure is summarized in \algref{alg:two point FlexRC}. 

\begin{algorithm}
    \caption{Two-point FlexRC reduction}
    \label{alg:two point FlexRC}
    \begin{algorithmic}[1]
        \Require $G, C \in \real^{N \times N}$, $B \in \real^{N \times p}$; frequency points $[0, s_2]$; port-reduction tolerance $\tol$. 
        \Ensure Reduced matrices $\tilde{G}$, $\tilde{C}$, and $\tilde{B}$. 
        \State Apply elimination as in \eqref{eq:SIP congruence transformation} to obtain $\tilde{G}_{\p}$, $\tilde{C}_{\p}$, and $\tilde{C}_{\c}$. 
        \State Compute an economic QR decomposition $\tilde{C}_{\c} = QR$. 
        \State Compute row norms $\rho_j = \twonorm{R(j, :)}$ for $j = 1,\ldots,p$. 
        \State Sort the row indices so that $\rho_{\pi_1} \leq \cdots \leq \rho_{\pi_p}$. 
        \State Choose the largest $m \in \set{0,1,\ldots,p}$ such that 
        \begin{align*}
            \sqrt{\rho_{\pi_1}^2 + \cdots + \rho_{\pi_m}^2}/\Frobenius{R} \leq \tol. 
        \end{align*}
        \State Set $\mathcal{J} = \set{\pi_1,\ldots,\pi_m}$, $\mathcal{I} = \set{1,\ldots,p}\setminus\mathcal{J}$, and $p_{\i} = \abs{\mathcal{I}}$. 
        \State Set $B_{\i} = Q(:,\mathcal{I})$ and $R_{\i} = R(\mathcal{I}, :)$. 
        \State Solve $(G_{\i} + s_2 C_{\i})X = B_{\i}$, and normalize each column of $X$ to obtain $\hat{V}_2$ and $D$ satisfying $\hat{V}_2D = X$. 
        \State Compute $T_2 = \tran{B_{\i}} \hat{V}_2$. 
        \State Compute $\hat{G}_{\i} = \tran{\hat{V}_2}G_{\i}\hat{V}_2$ and $\hat{C}_{\i} = \tran{\hat{V}_2}C_{\i}\hat{V}_2$. 
        \State Set $\tilde{G}_{\i} = \invtran{T_2}\hat{G}_{\i}\inv{T_2}$, $\tilde{C}_{\i} = \invtran{T_2}\hat{C}_{\i}\inv{T_2}$, and $\tilde{B}_{\i} = R_{\i}$. 
        \State Assemble the final reduced model: 
        \begin{align*}
            \tilde{G} = 
            \begin{bmatrix}
                \tilde{G}_{\p} & 0 \\ 
                0 & \tilde{G}_{\i}
            \end{bmatrix}, \qquad 
            \tilde{C} = 
            \begin{bmatrix}
                \tilde{C}_{\p} & \tran{R_{\i}} \\ 
                R_{\i} & \tilde{C}_{\i}
            \end{bmatrix}, \qquad 
            \tilde{B} = 
            \begin{bmatrix}
                I_p \\ 0
            \end{bmatrix}. 
        \end{align*}
    \end{algorithmic}
\end{algorithm}

\subsection{More frequency points}
When more than two frequency points are used, FlexRC projects the perturbed internal subsystem onto the rational Krylov subspace \cite{Ruhe1998practical} associated with the frequency sequence $[s_2,\ldots,s_q]$. Unlike the standard rational Arnoldi process in \cite{Ruhe1998practical}, FlexRC uses incomplete orthogonalization and constructs a nonorthogonal projection basis. This modification preserves the sparse coupling structure induced by $R_{\i}$. 

The first block $V_2$ has already been constructed in the two-point case. Starting from $V_2$, we construct the orthonormal blocks $V_3, \ldots, V_q$ one at a time, so that $[V_2, V_3, \ldots, V_q]$ spans the rational Krylov subspace associated with $[s_2,\ldots,s_q]$. For $k = 3,\ldots,q$, in the iteration associated with $s_k$, we first form the candidate block 
\begin{align*}
    W_k = \inv{(G_{\i} + s_k C_{\i})}C_{\i}V_{k-1}. 
\end{align*}
Since $V_2$ is not orthonormal in general, its contribution is removed using the orthogonal projector onto $\spanspace{V_2}$: 
\begin{align*}
    W_k \leftarrow W_k - V_2\inv{(\tran{V_2}V_2)}(\tran{V_2}W_k). 
\end{align*}
The block $W_k$ is then orthogonalized against the previously generated blocks $V_3, \ldots, V_{k-1}$ by Gram--Schmidt. Finally, as in the port-reduction step, an economic QR decomposition with deflation tolerance $\tol$ is applied to discard nearly dependent columns. The surviving columns form $V_k$, and the basis generated by the incomplete orthogonalization process is 
\begin{align*}
    V_{\i} = \left[V_2, V_3, \ldots, V_q \right]. 
\end{align*}
Thus the columns in $V_3, \ldots, V_q$ are orthonormal and orthogonal to $\spanspace{V_2}$, while $V_2$ itself is not orthonormal in general. This is why the process is described as incomplete orthogonalization. The procedure is summarized in \algref{alg:block rational Arnoldi process with incomplete orthogonalization}. 

\begin{algorithm}
    \caption{Block rational Arnoldi process with incomplete orthogonalization}
    \label{alg:block rational Arnoldi process with incomplete orthogonalization}
    \begin{algorithmic}[1]
        \Require $G_{\i}, C_{\i} \in \real^{(N-p) \times (N-p)}$; internal block $V_2$; frequency points $s_3,\ldots,s_q$; deflation tolerance $\tol$. 
        \Ensure Basis $V_{\i} = [V_2,V_3,\ldots,V_q]$. 
        \State Set $V_{\i} = V_2$. 
        \For {$k = 3,\ldots,q$}
            \State Compute $W = \inv{(G_{\i} + s_k C_{\i})}C_{\i}V_{k-1}$. 
            \State Remove the component in $V_2$: 
            \Statex \hspace{\algorithmicindent}$W \leftarrow W - V_2\inv{(\tran{V_2}V_2)}(\tran{V_2}W)$. 
            \For {$j = 3,\ldots,k-1$}
                \State Orthogonalize against the previous block: 
                \Statex \hspace{\algorithmicindent}$W \leftarrow W - V_j(\tran{V_j}W)$. 
            \EndFor
            \State Compute an economic QR decomposition $W = QR$, and let $n_k$ be the number of columns of $Q$. 
            \State Compute row norms $\rho_j = \twonorm{R(j, :)}$ for $j = 1,\ldots,n_k$. 
            \State Sort the row indices so that $\rho_{\pi_1} \leq \cdots \leq \rho_{\pi_{n_k}}$. 
            \State Choose the largest $m_k \in \set{0,1,\ldots,n_k}$ such that 
            \begin{align*}
                \sqrt{\rho_{\pi_1}^2 + \cdots + \rho_{\pi_{m_k}}^2}/\Frobenius{R} \leq \tol. 
            \end{align*}
            \State Set $\mathcal{J}_k = \set{\pi_1,\ldots,\pi_{m_k}}$ and $\mathcal{I}_k = \set{1,\ldots,n_k}\setminus\mathcal{J}_k$. 
            \State Set $V_k = Q(:,\mathcal{I}_k)$. 
            \If {$V_k$ is empty}
                \State Stop. 
            \EndIf
            \State Append $V_k$ to $V_{\i}$: $V_{\i} \leftarrow [V_{\i}, V_k]$. 
        \EndFor
    \end{algorithmic}
\end{algorithm}

We next apply a change of basis so that the reduced coupling matrix has the desired sparse block form. Write $V_{\i} = [\,V_2 \mid V_{\mathrm{tail}}\,]$ and compute 
\begin{align*}
    T_{\mathrm{tail}} = \tran{B_{\i}}V_{\mathrm{tail}}. 
\end{align*}
Since $\tran{B_{\i}} V_2 = I_{p_{\i}}$, replacing the tail block by 
\begin{align*}
    \tilde{V}_{\mathrm{tail}} = V_{\mathrm{tail}} - V_2 T_{\mathrm{tail}} 
\end{align*}
removes its component seen by $B_{\i}$:
\begin{align*}
    \tran{B_{\i}} \tilde{V}_{\mathrm{tail}} 
    = \tran{B_{\i}} V_{\mathrm{tail}} - \tran{B_{\i}} V_2 T_{\mathrm{tail}} 
    = T_{\mathrm{tail}} - T_{\mathrm{tail}} 
    = 0. 
\end{align*}
Therefore, with the new basis 
\begin{align*}
    \tilde{V}_{\i} = [\,V_2 \mid \tilde{V}_{\mathrm{tail}}\,], 
\end{align*}
it holds that 
\begin{align*}
    \tran{\tilde{V}_{\i}}B_{\i} = \begin{bmatrix} I_{p_{\i}} \\ 0 \end{bmatrix}. 
\end{align*}
After this change of basis, we project the perturbed internal subsystem \eqref{eq:perturbed internal subsystem} with $\tilde{V}_{\i}$. The internal reduced matrices are 
\begin{align*}
    \tilde{G}_{\i} = \tran{\tilde{V}_{\i}}G_{\i}\tilde{V}_{\i}, \qquad 
    \tilde{C}_{\i} = \tran{\tilde{V}_{\i}}C_{\i}\tilde{V}_{\i}. 
\end{align*}
The corresponding reduced coupling matrix is 
\begin{align*}
    \tilde{B}_{\i} = \tran{\tilde{V}_{\i}}B_{\i}R_{\i} = 
    \begin{bmatrix} 
        R_{\i} \\ 0 
    \end{bmatrix}. 
\end{align*}
Thus only the first $p_{\i}$ reduced internal coordinates are directly coupled to the port variables. 

In practical implementation, neither the scaled block $V_2 = \hat{V}_2\inv{T_2}$ nor the equivalent basis $\tilde{V}_{\i}$ is formed explicitly. The incomplete Arnoldi process is carried out with the normalized block $\hat{V}_2$, which spans the same subspace as $V_2$. After projecting with the computed basis $\hat{V}_{\i} = [\,\hat{V}_2 \mid V_{\mathrm{tail}}\,]$, the transformations induced by $T_2$ and $T_{\mathrm{tail}}$ are applied only to the reduced matrices. More precisely, let 
\begin{align*}
    T_{\i} =
    \begin{bmatrix}
        \inv{T_2} & -\inv{T_2}T_{\mathrm{tail}} \\ 
        0 & I
    \end{bmatrix}. 
\end{align*}
Denote $\hat{G}_{\i} = \tran{\hat{V}_{\i}}G_{\i}\hat{V}_{\i}$ and $\hat{C}_{\i} = \tran{\hat{V}_{\i}}C_{\i}\hat{V}_{\i}$. Then 
\begin{align*}
    \tilde{G}_{\i} = \tran{T_{\i}}\hat{G}_{\i}T_{\i}, \qquad 
    \tilde{C}_{\i} = \tran{T_{\i}}\hat{C}_{\i}T_{\i}. 
\end{align*}

With respect to $\tilde{V}_{\i} = [\,V_2 \mid \tilde{V}_{\mathrm{tail}}\,]$, we can write 
\begin{align*}
    \tilde{G}_{\i} =
    \begin{bmatrix}
        \tilde{G}_{2,2} & \tran{\tilde{G}_{\tail,2}} \\ 
        \tilde{G}_{\tail,2} & \tilde{G}_{\tail,\tail} 
    \end{bmatrix}, \qquad 
    \tilde{C}_{\i} =
    \begin{bmatrix}
        \tilde{C}_{2,2} & \tran{\tilde{C}_{\tail,2}} \\ 
        \tilde{C}_{\tail,2} & \tilde{C}_{\tail,\tail} 
    \end{bmatrix}, 
\end{align*}
where the blocks $\tilde{G}_{2,2}$ and $\tilde{C}_{2,2}$ have size $p_{\i} \times p_{\i}$. The following proposition shows that $\tilde{G}_{\tail,2}$ and $\tilde{C}_{\tail,2}$ are linearly dependent. 
\begin{proposition}
    \label{prop:block coupling relation}
    The off-diagonal blocks satisfy 
    \begin{align*}
        \tilde{G}_{\tail,2} = -s_2 \tilde{C}_{\tail,2}. 
    \end{align*}
\end{proposition}
The proof is given in \appref{app:proof block coupling relation}. Since $\tilde{G}_{\tail,2}$ is a scalar multiple of $\tilde{C}_{\tail,2}$, a QR factorization of $\tilde{C}_{\tail,2}$ can be used to sparsify both off-diagonal blocks. We compute a full QR decomposition 
\begin{align*}
    \tilde{C}_{\tail,2} = Q_{\tail} R_{\tail,2} 
\end{align*}
where $Q_{\tail}$ is a square orthogonal matrix, and $R_{\tail,2}$ has zeros below its main diagonal. We then apply the block diagonal congruence transformation induced by 
\begin{align*}
    P_{\tail} =
    \begin{bmatrix}
        I_{p_{\i}} & 0 \\ 
        0 & Q_{\tail}
    \end{bmatrix}
\end{align*}
to the internal reduced matrices. For notational simplicity, the transformed tail diagonal blocks are still denoted by $\tilde{G}_{\tail,\tail}$ and $\tilde{C}_{\tail,\tail}$. The transformed internal reduced matrices are 
\begin{align*}
    \tilde{G}_{\i} & =
    \begin{bmatrix}
        \tilde{G}_{2,2} & -s_2\tran{R_{\tail,2}} \\ 
        -s_2R_{\tail,2} & \tilde{G}_{\tail,\tail}
    \end{bmatrix}, \\ 
    \tilde{C}_{\i} & =
    \begin{bmatrix}
        \tilde{C}_{2,2} & \tran{R_{\tail,2}} \\ 
        R_{\tail,2} & \tilde{C}_{\tail,\tail}
    \end{bmatrix}. 
\end{align*}
The reduced coupling matrix remains unchanged under this transformation: 
\begin{align*}
    \tilde{B}_{\i} & = \tran{P_{\tail}}
    \begin{bmatrix}
        R_{\i} \\ 0
    \end{bmatrix}
    = \begin{bmatrix}
        R_{\i} \\ 0
    \end{bmatrix}. 
\end{align*}
For $q \geq 4$, similar sparsification operations can be further applied to the tail blocks $\tilde{G}_{\tail,\tail}$ and $\tilde{C}_{\tail,\tail}$, but this extension is not pursued because two or three points are usually sufficient in practical applications. 

Cascading the retained port-node subsystem with the projected internal subsystem gives the final reduced-order model: 
\begin{equation}
    \label{eq:reduced model of FlexRC}
    \begin{aligned}
    \tilde{G} & = \begin{bmatrix}
        \tilde{G}_{\p} & 0 & 0 \\ 
        0 & \tilde{G}_{2,2} & -s_2\tran{R_{\tail,2}} \\ 
        0 & -s_2R_{\tail,2} & \tilde{G}_{\tail,\tail}
    \end{bmatrix}, \\
    \tilde{C} & = \begin{bmatrix}
        \tilde{C}_{\p} & \tran{R_{\i}} & 0 \\ 
        R_{\i} & \tilde{C}_{2,2} & \tran{R_{\tail,2}} \\ 
        0 & R_{\tail,2} & \tilde{C}_{\tail,\tail}
    \end{bmatrix}, \\
    \tilde{B} & = \begin{bmatrix}
        I_p \\ 0 \\ 0
    \end{bmatrix}. 
    \end{aligned}
\end{equation}
The complete multi-point FlexRC procedure is summarized in \algref{alg:multi point FlexRC}. 
\begin{algorithm}
    \caption{Multi-point FlexRC reduction}
    \label{alg:multi point FlexRC}
    \begin{algorithmic}[1]
        \Require $G, C \in \real^{N \times N}$, $B \in \real^{N \times p}$; frequency points $[s_1=0,s_2,\ldots,s_q]$; port-reduction tolerance $\tol$. 
        \Ensure Reduced matrices $\tilde{G}$, $\tilde{C}$, and $\tilde{B}$. 
        \State Apply the elimination and port-reduction steps in \algref{alg:two point FlexRC} to obtain $G_{\i}$, $C_{\i}$, $B_{\i}$, $R_{\i}$, $\tilde{G}_{\p}$, $\tilde{C}_{\p}$, $\hat{V}_2$, and $T_2$. 
        \If {$q \leq 2$}
            \State Return the reduced model in \algref{alg:two point FlexRC}. 
        \EndIf
        \State Apply \algref{alg:block rational Arnoldi process with incomplete orthogonalization} with $\hat{V}_2$ to obtain $\hat{V}_{\i} = [\,\hat{V}_2 \mid V_{\mathrm{tail}}\,]$. 
        \State Project the internal subsystem with $\hat{V}_{\i}$ and apply the change of basis described above to obtain $\tilde{G}_{\i}$, $\tilde{C}_{\i}$, and $\tilde{B}_{\i}$. 
        \State Apply the block sparsification based on \propref{prop:block coupling relation} to obtain $R_{\tail,2}$ and the transformed internal blocks. 
        \State Assemble $\tilde{G}$, $\tilde{C}$, and $\tilde{B}$ as in \eqref{eq:reduced model of FlexRC}. 
    \end{algorithmic}
\end{algorithm}

\subsection{Discussion and sparsity control}
\label{sec:discussion and sparsity control}

The reduction process above generates a reduced model with a banded block structure, consistent with the structure obtained by TurboMOR-RC. Unlike a fixed reduction procedure, FlexRC leaves several choices to the user: the sequence of frequency points, the port-reduction tolerance, and the sparsity-control option. These adjustable components are the main source of the flexibility of FlexRC, allowing the reduced model to be tuned for accuracy, reduced order, and simulation efficiency. 

The first practical choice in FlexRC is the sequence of frequency points $\mathcal{S}$. For RC reduction, the frequency points are usually chosen as nonnegative real numbers, so that each matrix $G_{\i} + s_k C_{\i}$ is real and the whole reduction procedure only involves real arithmetic. The moment-matching effect of the frequency sequence is analyzed in \secref{sec:moment matching}. In our experience, two frequency points are usually sufficient to capture the time-domain behavior of practical RC networks. Using more than three points can substantially increase the reduced order, and the reduced model may even become less efficient than the original system. 

The second practical choice is the tolerance $\tol$. This parameter controls the trade-off between the accuracy of the reduced model and its simulation efficiency. In principle, the tolerance for port reduction and the tolerance for Arnoldi deflation can be set independently. In the numerical experiments of this paper, however, we use the same tolerance $\tol$ for both steps. The most conservative choice is $\tol = 0$, for which no direction is discarded in the port-reduction and deflation steps. In this case, if all frequency points are chosen as $0$, the resulting FlexRC model differs from the TurboMOR-RC model only by a nonsingular change of basis. By choosing an appropriate tolerance, FlexRC can significantly reduce the reduced order without noticeably affecting accuracy in many practical cases; an error analysis is given in \secref{sec:error analysis of port reduction}. However, an overly aggressive tolerance may degrade the accuracy. At present, an a priori strategy for selecting $\tol$ is not available, and this remains future work.

Finally, FlexRC also supports a sparsity-control variant. Eliminating too many local nodes in the first step may introduce severe fill-in and make the reduced model too dense. To control this effect, one can preserve more internal nodes in the first elimination step, following the spirit of SMP-RCR \cite{YXLYZAS2025SMPRCR}. Indeed, if $\mathcal{P}$ is the original port-node set, we choose an enlarged retained set $\mathcal{P}_{\mathrm{sc}}$ with $\mathcal{P} \subseteq \mathcal{P}_{\mathrm{sc}}$, where the additional nodes are auxiliary retained nodes rather than external ports. 

This sparsity-control variant keeps more variables before reduction, but it can substantially improve the sparsity of the final matrices by reducing fill-in in the initial elimination. In the port-reduction step, more directions can be discarded, and a sparse QR factorization often produces a sparser factor $R$. In our implementation, the enlarged retained set usually contains about two times the number of original ports, and a two-point model is often already sufficiently accurate with this strategy. 

\section{Analysis}
\label{sec:analysis}

\subsection{Passivity}

The passivity of the RC network is intrinsically linked to the symmetric positive semidefiniteness of the conductance and capacitance matrices. We use the following standard positive-real characterization from \cite{RS2011lyapunov}. 
\begin{theorem}
    Let $G, C \in \real^{N \times N}$ be symmetric positive semidefinite matrices, and let $B \in \real^{N \times p}$. If $\ker G \cap \ker C = \set{0}$, then the transfer function $H(s) = \tran{B}\inv{(G + s C)}B$ is positive real. Specifically, $H(s)$ satisfies the following conditions:
    \begin{enumerate}
        \item $H(s)$ is analytic in the open right half-plane $\complex^{+}$;
        \item $H(s) + \conjtran{H(s)} \succeq 0$ for all $s \in \complex^{+}$. 
    \end{enumerate}
\end{theorem}

Apart from port reduction, the overall FlexRC procedure can be summarized by two steps. The first step eliminates the coupling block of the conductance matrix by Gaussian elimination, which is equivalent to the congruence transformation with the matrix $M$ in \eqref{eq:SIP congruence transformation}. The second step is projection of the internal subsystem with the nonorthogonal basis $\tilde{V}_{\i}$, which is equivalent to a congruence-type projection of the transformed system with $\diag{I_p,\tilde{V}_{\i}}$. The additional block sparsification in the multi-point case is only an equivalent congruence change of basis inside the reduced internal subsystem. Therefore, if port reduction is not applied, the reduced matrices $\tilde{G}$ and $\tilde{C}$ remain symmetric positive semidefinite. 

The port-reduction step changes the reduced model only through the replacement of the coupling matrix $\tilde{C}_{\c}$ by $B_{\i}R_{\i}$. From the viewpoint of the reduced model obtained above, this replacement introduces a symmetric perturbation to the capacitance matrix $\tilde{C}$. Such a perturbation may theoretically destroy positive semidefiniteness. In practice, however, the reduced capacitance matrix before port reduction is usually positive definite, and a sufficiently small perturbation preserves positive definiteness by Weyl's lemma \cite{Book:GV1996}. In our numerical experiments, with the tolerances used for port reduction, no passivity violation has been observed. 

\subsection{Moment matching}
\label{sec:moment matching}

We now discuss the moment-matching properties. The analysis in this subsection assumes the exact case $\tol = 0$, so that no perturbation is introduced. In this case, the basis $\tilde{V}_{\i}$ spans the rational Krylov subspace of the internal subsystem associated with the frequency points $[s_2,\ldots,s_q]$. The moment matching of the full transfer function then follows from the internal moment matching and the Schur-complement relation. 

Let $H(s) = \tran{B}\inv{(G + sC)}B$ and $\tilde{H}(s) = \tran{\tilde{B}}\inv{(\tilde{G} + s\tilde{C})}\tilde{B}$ denote the transfer functions of the original and reduced systems, respectively. For any $s_0 \geq 0$, let $q(s_0)$ denote the number of occurrences of $s_0$ in the sequence $\mathcal{S} = [s_1, s_2, \ldots, s_q] = [0, s_2, \ldots, s_q]$. The following theorem characterizes the moment-matching property of FlexRC. 
\begin{theorem}
    \label{thm:moment matching at non-pole point of FlexRC}
    Suppose $G + s_0 C$ and $\tilde{G} + s_0 \tilde{C}$ are nonsingular. Then, for FlexRC with $\tol = 0$, $\tilde{H}(s)$ matches the first $2q(s_0)$ moments of $H(s)$ at $s = s_0$. 
\end{theorem}

In practical RC networks, the conductance matrix $G$ is often singular because some connected components of the resistor network may not contain a grounding resistor. In this case, $s_0 = 0$ may become a pole of the transfer function. Such cases have been considered in Lyapunov balancing to obtain a bounded $\mathcal{H}_{\infty}$-norm estimate \cite{RS2011lyapunov}, but they have not been addressed in the moment-matching analysis of elimination-based methods. For RC networks with singular $G$, if $\ker G \cap \ker C = \set{0}$, then the transfer functions $H(s)$ and $\tilde{H}(s)$ admit the following Laurent expansions \cite{RS2011lyapunov}: 
\begin{align*}
    H(s) = \frac{M_{-1}}{s} + \sum_{k=0}^{\infty} M_k s^k, \qquad \tilde{H}(s) = \frac{\tilde{M}_{-1}}{s} + \sum_{k=0}^{\infty} \tilde{M}_k s^k. 
\end{align*}
Since $s_1 = 0$ in the sequence $\mathcal{S}$, the number of occurrences $q(0) \geq 1$. The following theorem gives the Laurent moment matching of FlexRC. 
\begin{theorem}
    \label{thm:moment matching at pole point of FlexRC}
    Suppose the conductance matrix $G$ is singular and $\ker G \cap \ker C = \set{0}$. Then, for exact FlexRC with $\tol = 0$, $\tilde{H}(s)$ satisfies 
    \begin{align*}
        \tilde{M}_k = M_k, \quad k = -1, 0, \ldots, 2q(0)-3. 
    \end{align*}
\end{theorem}

\begin{proof}
    We prove \thmref{thm:moment matching at non-pole point of FlexRC} and \thmref{thm:moment matching at pole point of FlexRC} together. Let $H_{\i}(s)$ and $\tilde{H}_{\i}(s)$ be the transfer functions of the internal subsystem before and after reduction, respectively: 
    \begin{align*}
        H_{\i}(s)
        & = \tran{\tilde{C}_{\c}}\inv{(G_{\i} + sC_{\i})}\tilde{C}_{\c}, \\
        \tilde{H}_{\i}(s)
        & = \tran{\tilde{B}_{\i}}\inv{(\tilde{G}_{\i} + s\tilde{C}_{\i})}\tilde{B}_{\i}. 
    \end{align*}
    The transfer function of the full system can be written by the Schur complement as 
    \begin{align*}
        H(s)
        & = \inv{\left(\tilde{G}_{\p} + s\tilde{C}_{\p} - s^2H_{\i}(s)\right)}. 
    \end{align*}
    Similarly, the reduced model satisfies 
    \begin{align*}
        \tilde{H}(s)
        & = \inv{\left(\tilde{G}_{\p} + s\tilde{C}_{\p} - s^2\tilde{H}_{\i}(s)\right)}. 
    \end{align*}
    Therefore, 
    \begin{align*}
        H(s) - \tilde{H}(s)
        & = H(s)\left[s^2\bigl(H_{\i}(s) - \tilde{H}_{\i}(s)\bigr)\right]\tilde{H}(s). 
    \end{align*}
    Since $\tol = 0$, the basis $\tilde{V}_{\i}$ spans the rational Krylov subspace of the internal subsystem. 
    By \thmref{thm:moment-matching for RC reduction}, if $s_0$ appears $k$ times in the internal frequency sequence $[s_2, \ldots, s_q]$, then the internal transfer functions satisfy 
    \begin{align*}
        H_{\i}(s) - \tilde{H}_{\i}(s) = O((s - s_0)^{2k}). 
    \end{align*}

    We first consider the non-pole case in \thmref{thm:moment matching at non-pole point of FlexRC}. If $s_0 = 0$, then the frequency point $s_0$ is used once by the elimination step and appears $q(0)-1$ times in $[s_2, \ldots, s_q]$. Hence $H_{\i}(s) - \tilde{H}_{\i}(s) = O(s^{2q(0)-2})$, and 
    \begin{align*}
        H(s) - \tilde{H}(s)
        & = O(s^2) \cdot O(1) \cdot O(s^{2q(0)-2}) \cdot O(1) \\
        & = O(s^{2q(0)}). 
    \end{align*}
    If $s_0 \neq 0$, then $s_0$ appears $q(s_0)$ times in $[s_2, \ldots, s_q]$. Since $s^2 = s_0^2 + O(s - s_0)$ with $s_0^2 \neq 0$, we obtain 
    \begin{align*}
        H(s) - \tilde{H}(s)
        & = O(1) \cdot O(1) \cdot O\bigl((s - s_0)^{2q(s_0)}\bigr) \cdot O(1) \\
        & = O\bigl((s - s_0)^{2q(s_0)}\bigr). 
    \end{align*}
    Thus $\tilde{H}(s)$ matches the first $2q(s_0)$ moments of $H(s)$ at $s = s_0$. 

    We now consider the pole case in \thmref{thm:moment matching at pole point of FlexRC}. The same internal moment matching gives $H_{\i}(s) - \tilde{H}_{\i}(s) = O(s^{2q(0)-2})$. From the Laurent expansions, $H(s) = O(s^{-1})$ and $\tilde{H}(s) = O(s^{-1})$. Therefore 
    \begin{align*}
        H(s) - \tilde{H}(s)
        & = O(s^2) \cdot O(s^{-1}) \cdot O(s^{2q(0)-2}) \cdot O(s^{-1}) \\
        & = O(s^{2q(0)-2}). 
    \end{align*}
    Comparing Laurent coefficients gives $\tilde{M}_k = M_k$ for $k = -1, 0, \ldots, 2q(0) - 3$. 
\end{proof}

The result in \thmref{thm:moment matching at pole point of FlexRC} is not restricted to FlexRC. The same argument also applies to existing elimination-based RC reduction methods. Unlike the Lyapunov balancing methods in \cite{RS2011lyapunov}, which explicitly compute $M_{-1}$ to enforce Laurent moment matching, elimination-based methods obtain Laurent moment matching automatically through the elimination process. 

\subsection{Error analysis of port reduction}
\label{sec:error analysis of port reduction}

This subsection discusses how port reduction affects the accuracy of reduced models. The port-reduction step replaces the internal coupling matrix $\tilde{C}_{\c}$ by a low-rank approximation $B_{\i}R_{\i}$, which gives the perturbed internal subsystem 
\begin{align*}
    \perturbed{\Sigma}_{\i}: 
    \left\{\begin{aligned} 
        & (G_{\i} + s C_{\i})x_{\i} = (\tilde{C}_{\c} + \Delta)u_{\i}, \\ 
        & y_{\i} = \tran{(\tilde{C}_{\c} + \Delta)}x_{\i}. 
    \end{aligned}\right. 
\end{align*}
where the perturbation 
\begin{align*}
    \Delta \define B_{\i}R_{\i} - \tilde{C}_{\c}, \qquad 
    \Frobenius{\Delta} / \Frobenius{\tilde{C}_{\c}} \leq \tol. 
\end{align*}
To quantify the error between the original and perturbed internal subsystems, let $H_{\i}(\imaginary \omega)$ and $\perturbed{H}_{\i}(\imaginary \omega)$ be the transfer functions of $\Sigma_{\i}$ and $\perturbed{\Sigma}_{\i}$ at $\imaginary \omega$ ($\omega \in \real$), respectively: 
\begin{align*}
    H_{\i}(\imaginary \omega)
    & = \tran{\tilde{C}_{\c}}\inv{(G_{\i} + \imaginary \omega C_{\i})}\tilde{C}_{\c}, \\
    \perturbed{H}_{\i}(\imaginary \omega)
    & = \tran{(\tilde{C}_{\c} + \Delta)}
    \inv{(G_{\i} + \imaginary \omega C_{\i})}(\tilde{C}_{\c} + \Delta),
\end{align*}
and define the relative error as 
\begin{align*}
    \delta_{\i}(\imaginary \omega) \define \dfrac{\Frobenius{H_{\i}(\imaginary \omega) - \perturbed{H}_{\i}(\imaginary \omega)}}{\Frobenius{H_{\i}(\imaginary \omega)}}. 
\end{align*}
The following theorem gives a conservative error estimate for this relative error. 
\begin{theorem}
    \label{thm:internal perturbation bound}
    If $\Frobenius{\Delta} \leq \epsilon \Frobenius{\tilde{C}_{\c}}$, then 
    \begin{align*}
        \delta_{\i}(\imaginary \omega)
        \leq \sqrt{p}
        \dfrac{\twonorm{G_{\i} + \imaginary \omega C_{\i}}^2}{\lambda_{\min}(G_{\i})^2}
        \epsilon (2 + \epsilon), \quad \omega \in \real. 
    \end{align*}
\end{theorem}

The proof is given in \appref{app:proof internal perturbation bound}. 

Since the port-reduction criterion gives $\Frobenius{\Delta} \leq \tol \Frobenius{\tilde{C}_{\c}}$, the theorem implies that, for sufficiently small $\tol$, the internal transfer-function perturbation satisfies $\delta_{\i}(\imaginary \omega) = O(\tol)$. This is only an error estimate for the internal subsystem; the perturbation of the internal transfer function is further propagated to the full transfer function. The estimate here is conservative and should mainly be interpreted as a qualitative perturbation estimate controlled by $\tol$. Therefore, this inequality does not provide an a priori criterion for selecting $\tol$. Moreover, this bound estimates the internal transfer-function perturbation, while RC reduced models are usually used for transient simulation. In practice, the error in the subsequent transient simulation is usually much smaller than this inequality suggests. 

\section{Numerical experiments}
\label{sec:numerical experiments}

In this section, we evaluate FlexRC on practical RC networks in terms of reduction time, transient simulation time, frequency-point selection, and sparsity control. All experiments are run on one CPU core of a cluster node with an Intel Xeon Gold 6226R @2.90 GHz processor and $700$ GB memory, and all algorithms are implemented in MATLAB R2020b. 

\tableref{tab:RC original systems} summarizes the test cases. The first six examples are provided by our industrial partners, and the last two are IBM power-grid benchmarks from \cite{Nassif2008power}. The condition numbers are estimated by MATLAB \texttt{condest}. Except for \texttt{DLL\_net90}, the full conductance matrices $G$ are singular or nearly singular in floating-point arithmetic. However, elimination-based methods only require the internal conductance block $G_{\i}$ in the first elimination step, and \tableref{tab:RC original systems} shows that $G_{\i}$ is much better conditioned. 

FlexRC is compared with PRIMA \cite{OCP1998PRIMA}, SPRIM \cite{Freund2004SPRIM}, and TurboMOR-RC \cite{OT2016TurboMOR}. PRIMA and SPRIM are implemented in MATLAB. For TurboMOR-RC, we follow the implementation strategy of \cite{OT2016TurboMOR}: the algorithm is driven from MATLAB, while the key Householder operations in the QR decomposition and the products with Householder matrices are implemented in C and called through MEX files. 

\begin{table*}[!t]
    \centering
    \caption{Original RC network information.}
    \label{tab:RC original systems}
    \begin{tabular}{c|r|r|c|c|c|c|c}
        \hline
        Case & Nodes & Ports & Resistors & Capacitors & $\mathrm{nnz}(G+C)$ & $\kappa(G)$ & $\kappa(G_{\i})$ \\
        \hline
        \texttt{AAADC\_net64} & 11070 & 419 & 30065 & 4958 & 71200 & $5.98\times 10^{17}$ & $4.70\times 10^{7}$ \\
        \texttt{AAADC\_net76} & 10989 & 420 & 29897 & 4999 & 70783 & $6.10\times 10^{17}$ & $4.68\times 10^{7}$ \\
        \texttt{DAC\_net99} & 57959 & 467 & 84658 & 18238 & 227275 & $7.58\times 10^{19}$ & $1.01\times 10^{7}$ \\
        \texttt{DLL\_net90} & 75452 & 697 & 103381 & 31321 & 282212 & $5.00\times 10^{7}$ & $6.99\times 10^{6}$ \\
        \texttt{ADC\_10bit\_net220} & 94816 & 554 & 165883 & 3626 & 426582 & $1.36\times 10^{18}$ & $4.95\times 10^{7}$ \\
        \texttt{PLLM\_ANA\_net4} & 46746 & 500 & 77488 & 9509 & 201722 & $4.53\times 10^{19}$ & $3.11\times 10^{8}$ \\
        \texttt{ibmpg1t} & 25272 & 250 & 40801 & 10774 & 95934 & $2.14\times 10^{18}$ & $3.37\times 10^{5}$ \\
        \texttt{ibmpg2t} & 163787 & 1200 & 245163 & 36838 & 617035 & $1.15\times 10^{19}$ & $3.71\times 10^{5}$ \\
        \hline
    \end{tabular}
\end{table*}

For transient simulations, the industrial examples provide only the extracted $G$ and $C$ matrices, so we use a synthetic sinusoidal input with frequency $10^8$ Hz. Although post-layout nonlinear simulation is the target application, the nonlinear device models and complete netlists are unavailable; linear transient simulations of the extracted RC networks are sufficient for evaluating the efficiency of the reduced models. The IBM power-grid examples are connected to $1.8$ V power supplies and driven by switching current sources. For all test cases, the simulation interval is $[0, 10^{-8}]$ s, the time step is $10^{-11}$ s, and the sparse linear systems are solved by CHOLMOD \cite{CDHR2008CHOLMOD} from SuiteSparse \cite{Davis2006direct}. The reported relative error is defined as 
\begin{align*}
    \varepsilon_{\mathrm{rel}} \define \max_i \frac{\twonorm{y_i - \tilde{y}_i}}{\twonorm{y_i}}, 
\end{align*}
where $y_i$ and $\tilde{y}_i$ are the time-domain waveforms of the $i$-th output of the full and reduced systems. 

\subsection{Reduction time and reduced model order}
This subsection compares the model reduction time and reduced order. Here $q$ denotes the number of frequency points, and the frequency sequence of FlexRC is set to $\mathcal{S} = [s_1,\ldots,s_q] = [0,\ldots,0]$. Since PRIMA and SPRIM require a nonsingular matrix $G+s_0C$, their expansion point is set to $s_0 = 1$. For FlexRC, port reduction and Arnoldi deflation use the same tolerance, with default values $\tol = 10^{-3}$ for $q = 2$ and $\tol = 5\times 10^{-5}$ for $q = 3$, and with slight case-dependent adjustments. 

\tableref{tab:RC reduction time} reports the model reduction time and the resulting reduced order. The reduction-time speedup is measured relative to PRIMA. For $q = 1$, PRIMA, TurboMOR-RC, and FlexRC have the same reduced order because only the first moment block is retained. For $q = 2$ and $q = 3$, FlexRC usually produces smaller reduced models on the examples. For instance, when $q = 2$, the reduced orders of FlexRC are $777$, $1095$, $872$, and $729$ for \texttt{DAC\_net99}, \texttt{DLL\_net90}, \texttt{ADC\_10bit\_net220}, and \texttt{PLLM\_ANA\_net4}, respectively, compared with $934$, $1394$, $1108$, and $1000$ for PRIMA and TurboMOR-RC. This reduction is due to the low-rank approximation of $\tilde{C}_{\c}$ in the port-reduction step, which reduces the reduced order after the first elimination step. 

Except for SPRIM, the reduction times of the other three methods are of the same order. FlexRC is faster than PRIMA and TurboMOR-RC in the two-point case on all examples in this table, and remains comparable overall in the three-point case. 

\begin{table*}[!t]
    \centering
    %\scriptsize
    \setlength{\tabcolsep}{2.5pt}
    \caption{Reduction time and reduced model order for the different methods. All times are in seconds.}
    \label{tab:RC reduction time}
    \begin{tabular}{@{}c|c|cc|ccc|ccc|ccc@{}}
        \hline
        \multirow{2}{*}{Case} & \multirow{2}{*}{$q$} & \multicolumn{2}{c|}{PRIMA} & \multicolumn{3}{c|}{SPRIM} & \multicolumn{3}{c|}{TurboMOR-RC} & \multicolumn{3}{c}{FlexRC} \\
        & & Time & Order & Time & Order & Speedup & Time & Order & Speedup & Time & Order & Speedup \\
        \hline
        \multirow{3}{*}{\texttt{DAC\_net99}} & 1 & 2.386 & 467 & 2.870 & 934 & 0.83$\times$ & 0.855 & 467 & 2.79$\times$ & 1.004 & 467 & 2.38$\times$ \\
         & 2 & 3.843 & 934 & 6.054 & 1401 & 0.63$\times$ & 3.179 & 934 & 1.21$\times$ & \textbf{2.534} & \textbf{777} & \textbf{1.52$\times$} \\
         & 3 & 6.609 & 1401 & 10.24 & 1868 & 0.65$\times$ & 6.595 & 1401 & 1.00$\times$ & 6.842 & 1351 & 0.97$\times$ \\
        \hline
        \multirow{3}{*}{\texttt{DLL\_net90}} & 1 & 4.695 & 697 & 7.525 & 1394 & 0.62$\times$ & 2.051 & 697 & 2.29$\times$ & 2.507 & 697 & 1.87$\times$ \\
         & 2 & 10.65 & 1394 & 16.50 & 2091 & 0.65$\times$ & 8.158 & 1394 & 1.31$\times$ & \textbf{6.780} & \textbf{1095} & \textbf{1.57$\times$} \\
         & 3 & 17.93 & 2091 & 25.63 & 2788 & 0.70$\times$ & 17.06 & 2091 & 1.05$\times$ & 18.86 & 1877 & 0.95$\times$ \\
        \hline
        \multirow{3}{*}{\texttt{ADC\_10bit\_net220}} & 1 & 4.586 & 554 & 7.611 & 1108 & 0.60$\times$ & 2.397 & 554 & 1.91$\times$ & 3.319 & 554 & 1.38$\times$ \\
         & 2 & 10.93 & 1108 & 16.56 & 1662 & 0.66$\times$ & 9.051 & 1108 & 1.21$\times$ & \textbf{6.919} & \textbf{872} & \textbf{1.58$\times$} \\
         & 3 & 18.10 & 1662 & 27.15 & 2216 & 0.67$\times$ & 18.69 & 1662 & 0.97$\times$ & 19.79 & 1390 & 0.91$\times$ \\
        \hline
        \multirow{3}{*}{\texttt{PLLM\_ANA\_net4}} & 1 & 1.508 & 500 & 2.647 & 1000 & 0.57$\times$ & 0.893 & 500 & 1.69$\times$ & 0.953 & 500 & 1.58$\times$ \\
         & 2 & 4.158 & 1000 & 6.330 & 1500 & 0.66$\times$ & 3.680 & 1000 & 1.13$\times$ & \textbf{2.074} & \textbf{729} & \textbf{2.01$\times$} \\
         & 3 & 7.092 & 1500 & 10.18 & 2000 & 0.70$\times$ & 7.514 & 1500 & 0.94$\times$ & 7.178 & 1320 & 0.99$\times$ \\
        \hline
        \multirow{3}{*}{\texttt{ibmpg1t}} & 1 & 0.434 & 250 & 0.637 & 500 & 0.68$\times$ & 0.260 & 250 & 1.67$\times$ & 0.294 & 250 & 1.48$\times$ \\
         & 2 & 0.942 & 500 & 1.282 & 750 & 0.73$\times$ & 0.815 & 500 & 1.16$\times$ & 0.739 & 500 & 1.27$\times$ \\
         & 3 & 1.483 & 750 & 2.131 & 1000 & 0.70$\times$ & 1.597 & 750 & 0.93$\times$ & 1.505 & 750 & 0.99$\times$ \\
        \hline
        \multirow{3}{*}{\texttt{ibmpg2t}} & 1 & 30.86 & 1200 & 42.44 & 2400 & 0.73$\times$ & 22.62 & 1200 & 1.36$\times$ & 25.39 & 1200 & 1.22$\times$ \\
         & 2 & 70.36 & 2400 & 92.76 & 3600 & 0.76$\times$ & 77.25 & 2400 & 0.91$\times$ & 58.95 & 2400 & 1.19$\times$ \\
         & 3 & 116.2 & 3600 & 154.0 & 4800 & 0.75$\times$ & 138.8 & 3600 & 0.84$\times$ & 124.0 & 2870 & 0.94$\times$ \\
        \hline
    \end{tabular}
\end{table*}

\subsection{Simulation time}

This subsection shows that the reduced models generated by FlexRC are efficient in subsequent transient simulation. \tableref{tab:RC simulation time} reports the transient simulation time, speedup, and relative error of the reduced models. The speedup here is computed with respect to the full-order transient simulation. 

The simulation results show the main benefit of FlexRC. In the two-point case, FlexRC gives the fastest transient simulations on the four industrial examples in this table, with speedups of $1.28\times$, $1.21\times$, $3.44\times$, and $1.96\times$ on \texttt{DAC\_net99}, \texttt{DLL\_net90}, \texttt{ADC\_10bit\_net220}, and \texttt{PLLM\_ANA\_net4}, respectively. These speedups are higher than those of PRIMA, SPRIM, and TurboMOR-RC, while the relative errors remain acceptable. In the three-point case, the larger reduced models are slower than the full-order simulation for three of the four examples. For the two IBM power-grid examples, using three points is still beneficial because the errors of the one-point and two-point models are relatively large. 

The accuracy of PRIMA and SPRIM may be affected by the ill-conditioning of the shifted linear systems on some examples. For $q = 2$, FlexRC gives errors close to those of TurboMOR-RC, while producing reduced models that are faster to simulate on the examples. For $q = 3$, FlexRC can lead to slightly larger errors than TurboMOR-RC in some cases, but these errors are nevertheless sufficient for the transient simulations considered here. Meanwhile, the additional accuracy of the three-point models is not always useful for improving transient simulation efficiency, because the larger reduced models may become more expensive to solve than the original systems. 

These results illustrate the benefit of port reduction. Compared with TurboMOR-RC, FlexRC applies a low-rank approximation to $\tilde{C}_{\c}$ to reduce the reduced order, which improves the efficiency of the subsequent transient simulation. The tolerance therefore provides a practical way to balance accuracy and simulation cost. 

\begin{table*}[!t]
    \centering
    %\scriptsize
    \setlength{\tabcolsep}{2pt}
    \caption{Simulation time and relative error for the reduced models obtained with the different methods. All times are in seconds.}
    \label{tab:RC simulation time}
    \begin{tabular}{@{}c|c|c|ccc|ccc|ccc|ccc@{}}
        \hline
        \multirow{2}{*}{Case} & \multirow{2}{*}{Full time} & \multirow{2}{*}{$q$} & \multicolumn{3}{c|}{PRIMA} & \multicolumn{3}{c|}{SPRIM} & \multicolumn{3}{c|}{TurboMOR-RC} & \multicolumn{3}{c}{FlexRC} \\
        & & & Time & Speedup & Error & Time & Speedup & Error & Time & Speedup & Error & Time & Speedup & Error \\
        \hline
        \multirow{3}{*}{\texttt{DAC\_net99}} & \multirow{3}{*}{31.49} & 1 & 8.999 & 3.50$\times$ & 1.29e-02 & 128.0 & 0.25$\times$ & 1.21e-02 & 9.931 & 3.17$\times$ & 1.29e-02 & 9.660 & 3.26$\times$ & 1.29e-02 \\
         &  & 2 & 48.77 & 0.65$\times$ & 1.21e-06 & 306.8 & 0.10$\times$ & 1.05e-06 & 36.99 & 0.85$\times$ & 1.35e-06 & \textbf{24.57} & \textbf{1.28$\times$} & 4.79e-06 \\
         &  & 3 & 116.5 & 0.27$\times$ & 1.83e-07 & 323.4 & 0.10$\times$ & 5.96e-08 & 70.66 & 0.45$\times$ & 8.08e-10 & 66.82 & 0.47$\times$ & 3.68e-08 \\
        \hline
        \multirow{3}{*}{\texttt{DLL\_net90}} & \multirow{3}{*}{47.42} & 1 & 22.46 & 2.11$\times$ & 8.88e-02 & 376.0 & 0.13$\times$ & 7.74e-02 & 14.57 & 3.26$\times$ & 8.88e-02 & 14.82 & 3.20$\times$ & 8.88e-02 \\
         &  & 2 & 99.21 & 0.48$\times$ & 1.13e-04 & 900.5 & 0.05$\times$ & 1.10e-04 & 70.88 & 0.67$\times$ & 1.14e-04 & \textbf{39.14} & \textbf{1.21$\times$} & 1.47e-04 \\
         &  & 3 & 375.4 & 0.13$\times$ & 6.80e-07 & 1067.3 & 0.04$\times$ & 6.46e-07 & 166.0 & 0.29$\times$ & 9.84e-07 & 135.9 & 0.35$\times$ & 2.47e-06 \\
        \hline
        \multirow{3}{*}{\texttt{ADC\_10bit\_net220}} & \multirow{3}{*}{103.3} & 1 & 14.87 & 6.95$\times$ & 2.54e-02 & 201.9 & 0.51$\times$ & 2.45e-02 & 13.39 & 7.72$\times$ & 2.54e-02 & 13.24 & 7.81$\times$ & 2.54e-02 \\
         &  & 2 & 69.14 & 1.49$\times$ & 2.80e-05 & 471.5 & 0.22$\times$ & 1.61e-05 & 50.71 & 2.04$\times$ & 2.76e-05 & \textbf{30.08} & \textbf{3.44$\times$} & 2.41e-05 \\
         &  & 3 & 185.8 & 0.56$\times$ & 1.18e-05 & 605.4 & 0.17$\times$ & 2.60e-06 & 91.59 & 1.13$\times$ & 7.48e-08 & 67.00 & 1.54$\times$ & 7.60e-08 \\
        \hline
        \multirow{3}{*}{\texttt{PLLM\_ANA\_net4}} & \multirow{3}{*}{39.50} & 1 & 10.99 & 3.60$\times$ & 4.75e-03 & 148.5 & 0.27$\times$ & 1.39e-03 & 10.05 & 3.93$\times$ & 2.44e-03 & 10.65 & 3.71$\times$ & 2.44e-03 \\
         &  & 2 & 55.61 & 0.71$\times$ & 7.72e-06 & 343.8 & 0.11$\times$ & 4.77e-06 & 42.09 & 0.94$\times$ & 3.04e-07 & \textbf{20.17} & \textbf{1.96$\times$} & 5.05e-07 \\
         &  & 3 & 133.8 & 0.30$\times$ & 3.79e-06 & 389.4 & 0.10$\times$ & 3.11e-06 & 73.72 & 0.54$\times$ & 7.66e-10 & 66.70 & 0.59$\times$ & 1.50e-09 \\
        \hline
        \multirow{3}{*}{\texttt{ibmpg1t}} & \multirow{3}{*}{18.86} & 1 & 0.989 & 19.07$\times$ & 7.83e-01 & 16.62 & 1.14$\times$ & 6.76e-01 & 0.962 & 19.61$\times$ & 7.83e-01 & 1.032 & 18.27$\times$ & 7.83e-01 \\
         &  & 2 & 3.326 & 5.67$\times$ & 3.11e-01 & 58.39 & 0.32$\times$ & 2.97e-01 & 5.079 & 3.71$\times$ & 3.11e-01 & 3.806 & 4.96$\times$ & 3.11e-01 \\
         &  & 3 & 6.989 & 2.70$\times$ & 1.10e-02 & 79.58 & 0.24$\times$ & 1.10e-02 & 9.587 & 1.97$\times$ & 1.10e-02 & \textbf{8.128} & \textbf{2.32$\times$} & 1.10e-02 \\
        \hline
        \multirow{3}{*}{\texttt{ibmpg2t}} & \multirow{3}{*}{367.6} & 1 & 31.89 & 11.53$\times$ & 6.87e-01 & 803.2 & 0.46$\times$ & 5.76e-01 & 28.96 & 12.69$\times$ & 6.87e-01 & 29.88 & 12.30$\times$ & 6.87e-01 \\
         &  & 2 & 150.9 & 2.44$\times$ & 6.55e-02 & 1504.8 & 0.24$\times$ & 6.20e-02 & 208.9 & 1.76$\times$ & 6.55e-02 & 159.2 & 2.31$\times$ & 6.55e-02 \\
         &  & 3 & 372.9 & 0.99$\times$ & 1.66e-03 & 3051.7 & 0.12$\times$ & 1.49e-04 & 394.6 & 0.93$\times$ & 3.10e-04 & \textbf{225.8} & \textbf{1.63$\times$} & 3.10e-04 \\
        \hline
    \end{tabular}
\end{table*}

To give a more direct view of the transient behavior, \figref{fig:RC ibmpg2t response error} compares the output waveform of the original \texttt{ibmpg2t} model against those of TurboMOR-RC and FlexRC. The right panel plots the error $y_1(t)-\tilde{y}_1(t)$. The two reduced models capture the dominant waveform accurately, and the error plot shows that the remaining discrepancy is small compared with the signal magnitude over the full transient interval. 

\begin{figure*}[!t]
    \centering
    \begin{minipage}{0.45\textwidth}
        \centering
        \includegraphics[width=\linewidth]{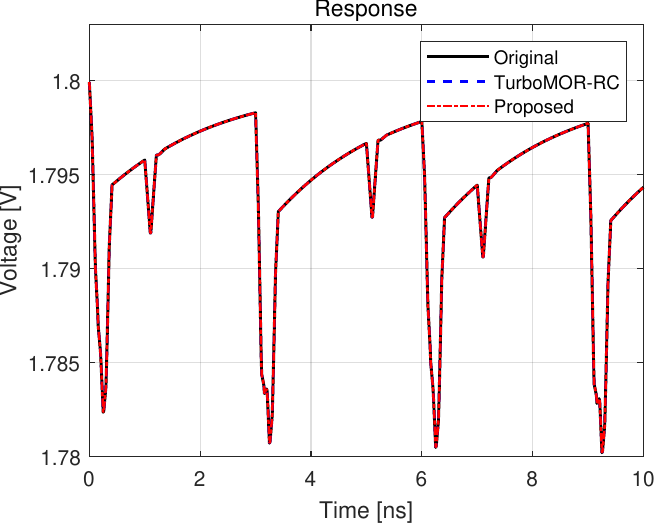}
    \end{minipage}
    \hfill
    \begin{minipage}{0.45\textwidth}
        \centering
        \includegraphics[width=\linewidth]{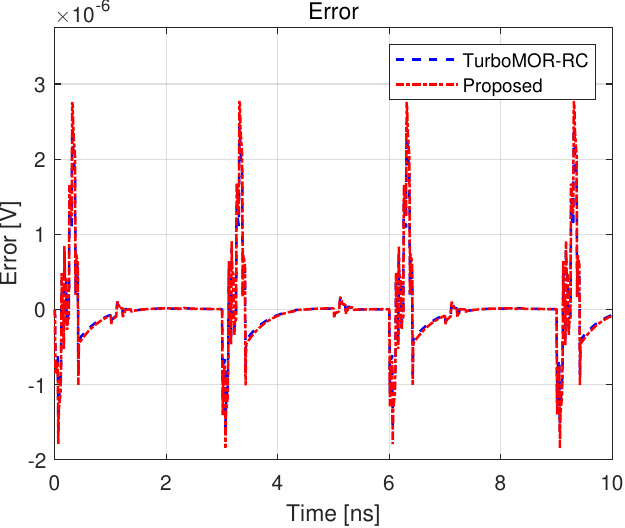}
    \end{minipage}
    \caption{Transient response and signed error for the first output of \texttt{ibmpg2t} with $p = 1200$. The left panel compares the original system, TurboMOR-RC, and FlexRC, while the right panel shows $y_1(t)-\tilde{y}_1(t)$ for the reduced models. The reduced models are constructed with the frequency points $[0,0,0]$.}
    \label{fig:RC ibmpg2t response error}
\end{figure*}

\subsection{Frequency points}

This subsection examines the effect of different frequency points. The goal is to verify that choosing suitable frequency points can improve transient accuracy without noticeably changing the reduced order or simulation time. In the previous experiments, the frequency sequence of FlexRC is set to $\mathcal{S} = [0,\ldots,0]$. Unlike TurboMOR-RC, whose frequency points are fixed at $s = 0$, FlexRC can use other frequency points in the reduction framework to improve accuracy. The results for \texttt{ibmpg1t} and \texttt{ibmpg2t} are reported in \tableref{tab:RC frequency point comparison ibmpg1t} and \tableref{tab:RC frequency point comparison ibmpg2t}, respectively. 

The results show that choosing suitable frequency points can effectively improve accuracy, while the reduction time and simulation time remain almost unchanged. For \texttt{ibmpg1t}, replacing $[0,0,0]$ by $[0,10^{10},10^{10}]$ reduces the three-point relative error from $1.10\times 10^{-2}$ to $6.38\times 10^{-4}$. For \texttt{ibmpg2t}, the corresponding error is reduced from $3.10\times 10^{-4}$ to $3.53\times 10^{-6}$. These results illustrate the main advantage of multi-point moment matching: the first frequency point $s_1 = 0$ preserves the DC behavior, while additional nonzero points can be selected to improve transient accuracy without noticeably increasing the simulation time. 

\begin{table*}[!t]
    \centering
    \caption{Comparison of different frequency points for \texttt{ibmpg1t}. All times are in seconds.}
    \label{tab:RC frequency point comparison ibmpg1t}
    \begin{tabular}{c|c|c|c|c|c}
        \hline
        Points & Frequency points & Model order & Reduction time & Simulation time & Relative error \\
        \hline
        2 & $[0,0]$ & 500 & 0.739 & 3.806 & 3.11e-01 \\
        2 & $[0,10^9]$ & 500 & 0.753 & 3.788 & 1.38e-01 \\
        2 & $[0,10^{10}]$ & 500 & 0.753 & 3.812 & 1.31e-02 \\
        \hline
        3 & $[0,0,0]$ & 750 & 1.505 & 8.128 & 1.10e-02 \\
        3 & $[0,10^9,10^9]$ & 750 & 1.544 & 7.966 & 1.37e-03 \\
        3 & $[0,10^{10},10^{10}]$ & 750 & 1.554 & 7.965 & 6.38e-04 \\
        \hline
    \end{tabular}
\end{table*}

\begin{table*}[!t]
    \centering
    \caption{Comparison of different frequency points for \texttt{ibmpg2t}. All times are in seconds.}
    \label{tab:RC frequency point comparison ibmpg2t}
    \begin{tabular}{c|c|c|c|c|c}
        \hline
        Points & Frequency points & Model order & Reduction time & Simulation time & Relative error \\
        \hline
        2 & $[0,0]$ & 2400 & 58.95 & 159.2 & 6.55e-02 \\
        2 & $[0,10^9]$ & 2400 & 61.10 & 160.8 & 2.70e-02 \\
        2 & $[0,10^{10}]$ & 2400 & 61.78 & 162.8 & 3.78e-03 \\
        \hline
        3 & $[0,0,0]$ & 2870 & 124.0 & 225.8 & 3.10e-04 \\
        3 & $[0,10^9,10^9]$ & 2872 & 126.4 & 223.9 & 5.34e-05 \\
        3 & $[0,10^{10},10^{10}]$ & 2874 & 122.9 & 227.4 & 3.53e-06 \\
        \hline
    \end{tabular}
\end{table*}

\subsection{Sparsity control}

The last experiment focuses on the sparsity-control variant of FlexRC discussed in \secref{sec:discussion and sparsity control}. The purpose is to show that sparsity control can generate sparser reduced models for transient simulation. Compared with the previously considered six examples, \texttt{AAADC\_net64} and \texttt{AAADC\_net76} have a larger fraction of port nodes among all nodes. For these cases, the reduced models generated by the previous methods can become inefficient in transient simulation. FlexRC can address this issue by sparsity control, and the resulting variant is denoted by ``FlexRC-SC''. 

\tableref{tab:RC sparsity control reduction time} reports the reduction time and reduced order, and \tableref{tab:RC sparsity control simulation time} reports the corresponding transient simulation time. FlexRC-SC generates a larger reduced model and requires a slightly longer reduction time, but it can significantly reduce the transient simulation time. For $q = 2$, the simulation time decreases from $12.75$ s to $3.088$ s on \texttt{AAADC\_net64}, and from $13.29$ s to $3.358$ s on \texttt{AAADC\_net76}. The relative error is also reduced from $4.20\times 10^{-2}$ to $1.80\times 10^{-5}$ on \texttt{AAADC\_net64}, and from $1.13\times 10^{-1}$ to $1.62\times 10^{-7}$ on \texttt{AAADC\_net76}. With sparsity control, the two-point model is already sufficiently accurate for these two examples. 

\begin{table*}[!t]
    \centering
    %\scriptsize
    \setlength{\tabcolsep}{3pt}
    \caption{Reduction time and reduced model order for the sparsity-control experiment. All times are in seconds.}
    \label{tab:RC sparsity control reduction time}
    \begin{tabular}{@{}c|c|cc|ccc|ccc|ccc|ccc@{}}
        \hline
        \multirow{2}{*}{Case} & \multirow{2}{*}{$q$} & \multicolumn{2}{c|}{PRIMA} & \multicolumn{3}{c|}{SPRIM} & \multicolumn{3}{c|}{TurboMOR-RC} & \multicolumn{3}{c|}{FlexRC} & \multicolumn{3}{c}{FlexRC-SC} \\
        & & Time & Order & Time & Order & Speedup & Time & Order & Speedup & Time & Order & Speedup & Time & Order & Speedup \\
        \hline
        \multirow{3}{*}{\texttt{AAADC\_net64}} & 1 & 0.318 & 419 & 0.490 & 838 & 0.65$\times$ & 0.162 & 419 & 1.97$\times$ & 0.179 & 419 & 1.77$\times$ & 0.270 & 838 & 1.18$\times$ \\
         & 2 & 0.637 & 838 & 0.995 & 1257 & 0.64$\times$ & 0.515 & 838 & 1.24$\times$ & 0.523 & 814 & 1.22$\times$ & 1.074 & 1256 & 0.59$\times$ \\
         & 3 & 1.071 & 1257 & 1.482 & 1676 & 0.72$\times$ & 1.024 & 1257 & 1.05$\times$ & 1.000 & 867 & 1.07$\times$ & 1.538 & 1266 & 0.70$\times$ \\
        \hline
        \multirow{3}{*}{\texttt{AAADC\_net76}} & 1 & 0.289 & 420 & 0.479 & 840 & 0.60$\times$ & 0.151 & 420 & 1.91$\times$ & 0.164 & 420 & 1.76$\times$ & 0.272 & 840 & 1.06$\times$ \\
         & 2 & 0.649 & 840 & 0.967 & 1260 & 0.67$\times$ & 0.513 & 840 & 1.27$\times$ & 0.494 & 811 & 1.31$\times$ & 1.121 & 1285 & 0.58$\times$ \\
         & 3 & 1.078 & 1260 & 1.485 & 1680 & 0.73$\times$ & 1.008 & 1260 & 1.07$\times$ & 0.944 & 853 & 1.14$\times$ & 1.677 & 1540 & 0.64$\times$ \\
        \hline
    \end{tabular}
\end{table*}

\begin{table*}[!t]
    \centering
    %\scriptsize
    %\small
    \setlength{\tabcolsep}{2.5pt}
    \caption{Simulation time and relative error for the sparsity-control experiment. All times are in seconds.}
    \label{tab:RC sparsity control simulation time}
    \begin{tabular}{@{}c|c|c|ccc|ccc|ccc|ccc|ccc@{}}
        \hline
        \multirow{2}{*}{Case} & \multirow{2}{*}{Full time} & \multirow{2}{*}{$q$} & \multicolumn{3}{c|}{PRIMA} & \multicolumn{3}{c|}{SPRIM} & \multicolumn{3}{c|}{TurboMOR-RC} & \multicolumn{3}{c|}{FlexRC} & \multicolumn{3}{c}{FlexRC-SC} \\
        & & & Time & Spd. & Error & Time & Spd. & Error & Time & Spd. & Error & Time & Spd. & Error & Time & Spd. & Error \\
        \hline
        \multirow{3}{*}{\texttt{AAADC\_net64}} & \multirow{3}{*}{5.705} & 1 & 5.851 & 0.97$\times$ & 4.78e+00 & 84.68 & 0.07$\times$ & 4.78e+00 & 2.977 & 1.92$\times$ & 4.78e+00 & 2.819 & 2.02$\times$ & 4.78e+00 & 1.058 & 5.39$\times$ & 1.02e-01 \\
         & & 2 & 29.95 & 0.19$\times$ & 2.62e-01 & 221.7 & 0.03$\times$ & 2.13e-01 & 15.55 & 0.37$\times$ & 6.14e-02 & 12.75 & 0.45$\times$ & 4.20e-02 & \textbf{3.088} & \textbf{1.85$\times$} & \textbf{1.80e-05} \\
         & & 3 & 65.57 & 0.09$\times$ & 8.90e-04 & 360.1 & 0.02$\times$ & 3.19e-04 & 35.11 & 0.16$\times$ & 1.14e-05 & 14.71 & 0.39$\times$ & 2.02e-05 & 3.126 & 1.83$\times$ & 5.72e-08 \\
        \hline
        \multirow{3}{*}{\texttt{AAADC\_net76}} & \multirow{3}{*}{5.683} & 1 & 6.196 & 0.92$\times$ & 4.80e+00 & 91.06 & 0.06$\times$ & 4.80e+00 & 3.128 & 1.82$\times$ & 4.80e+00 & 3.149 & 1.80$\times$ & 4.80e+00 & 1.133 & 5.02$\times$ & 5.31e-03 \\
         & & 2 & 27.86 & 0.20$\times$ & 3.50e-01 & 241.3 & 0.02$\times$ & 4.18e-01 & 18.36 & 0.31$\times$ & 3.30e-02 & 13.29 & 0.43$\times$ & 1.13e-01 & \textbf{3.358} & \textbf{1.69$\times$} & \textbf{1.62e-07} \\
         & & 3 & 67.23 & 0.08$\times$ & 3.16e-04 & 374.5 & 0.02$\times$ & 1.44e-04 & 36.60 & 0.16$\times$ & 2.41e-06 & 14.73 & 0.39$\times$ & 3.04e-04 & 15.70 & 0.36$\times$ & 1.26e-08 \\
        \hline
    \end{tabular}
\end{table*}

\figref{fig:RC AAADC net64 spy} and \figref{fig:RC AAADC net76 spy} compare the sparsity patterns of $\tilde{G}+\tilde{C}$ for the default FlexRC models and FlexRC-SC in the two-point case. The default reduced matrices are dense, while the matrices with sparsity control are much sparser. The number of nonzeros decreases from $300682$ to $47692$ for \texttt{AAADC\_net64}, and from $297312$ to $48603$ for \texttt{AAADC\_net76}, which explains the faster transient simulation observed in \tableref{tab:RC sparsity control simulation time}. 

\begin{figure*}[!t]
    \centering
    \begin{minipage}{0.36\textwidth}
        \centering
        \includegraphics[width=\linewidth]{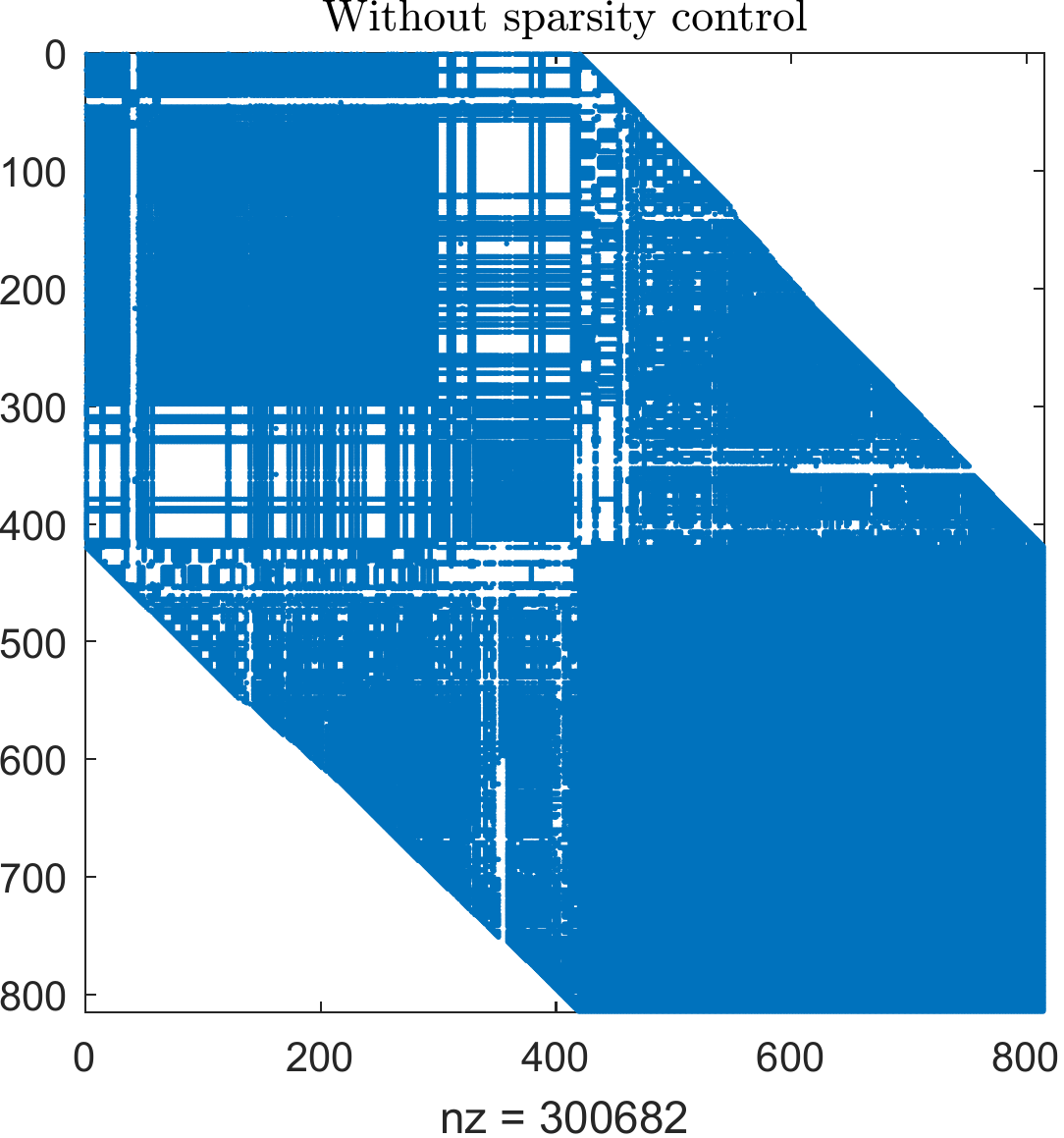}
    \end{minipage}
    \hspace{0.03\textwidth}
    \begin{minipage}{0.36\textwidth}
        \centering
        \includegraphics[width=\linewidth]{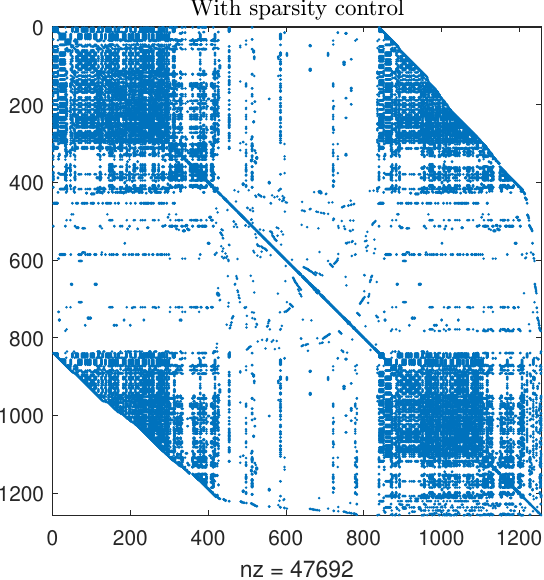}
    \end{minipage}
    \caption{Sparsity patterns of $\tilde{G}+\tilde{C}$ for the two-point reduced models of \texttt{AAADC\_net64}. The left panel shows the default FlexRC model with $300682$ nonzeros, and the right panel shows FlexRC-SC with $47692$ nonzeros.}
    \label{fig:RC AAADC net64 spy}
\end{figure*}

\begin{figure*}[!t]
    \centering
    \begin{minipage}{0.36\textwidth}
        \centering
        \includegraphics[width=\linewidth]{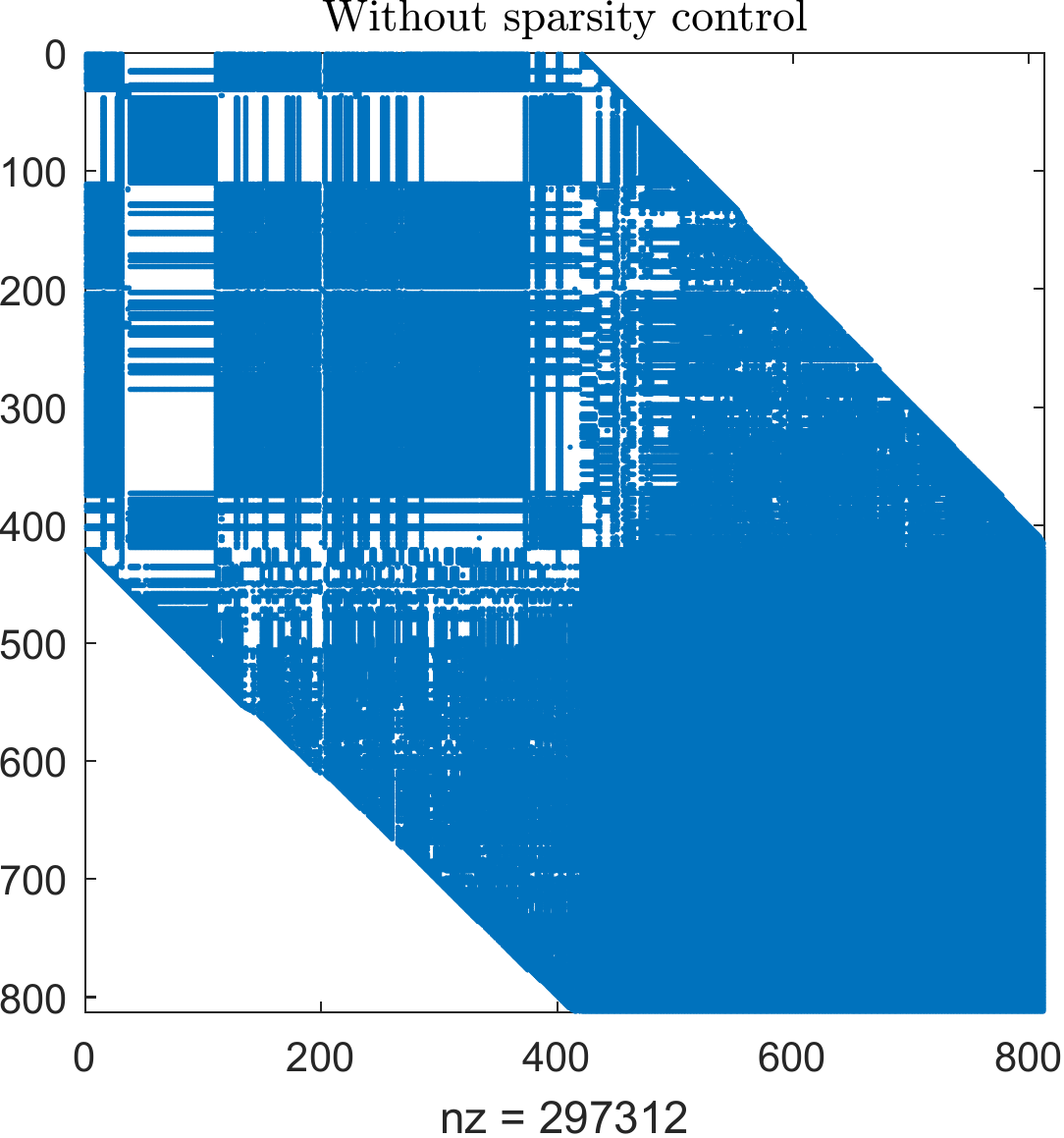}
    \end{minipage}
    \hspace{0.03\textwidth}
    \begin{minipage}{0.36\textwidth}
        \centering
        \includegraphics[width=\linewidth]{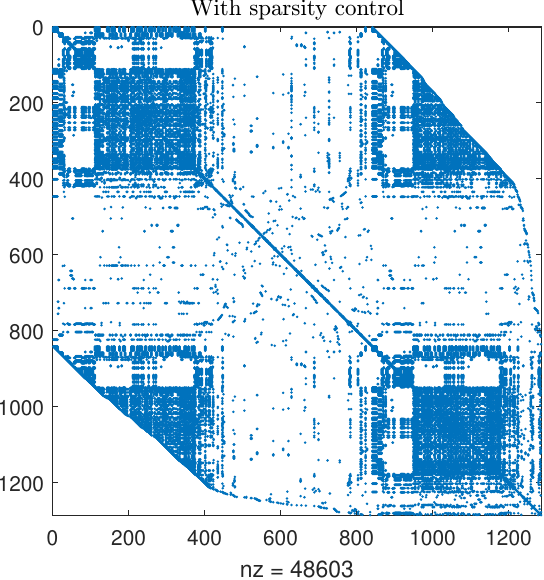}
    \end{minipage}
    \caption{Sparsity patterns of $\tilde{G}+\tilde{C}$ for the two-point reduced models of \texttt{AAADC\_net76}. The left panel shows the default FlexRC model with $297312$ nonzeros, and the right panel shows FlexRC-SC with $48603$ nonzeros.}
    \label{fig:RC AAADC net76 spy}
\end{figure*}

With sparsity control, less fill-in is introduced in the elimination step. However, if TurboMOR-RC is applied directly, the reduced model can have a large order, and the blocks other than the upper-left Schur-complement block can still be dense. The port-reduction step in FlexRC can not only reduce the extra order introduced by sparsity control but also preserve the sparsity. As a result, FlexRC-SC generates a reduced model that is more efficient in the subsequent transient simulation. 

\section{Conclusion}
\label{sec:conclusion}

In this paper, we proposed FlexRC, a flexible multi-point model order reduction method for large-scale RC networks with many ports. FlexRC allows the user to specify the frequency points to improve accuracy for the target response. A key novelty of FlexRC is the port-reduction technique: by adjusting the port-reduction tolerance, FlexRC can reduce the model order and improve the efficiency of subsequent transient simulation with only a slight accuracy loss. We discussed passivity under port-reduction perturbations. For the exact reduction, passivity follows from the congruence structure of the construction. When port reduction is applied, the reduced capacitance matrix is symmetrically perturbed, and passivity is preserved if this perturbation is sufficiently small. We also analyzed moment matching and provided a conservative error estimate for port reduction. Beyond the standard moment-matching analysis, we also considered the case where the conductance matrix is singular and proved a Laurent moment matching result. Numerical experiments on industrial RC examples and IBM power-grid examples show that the flexible framework of FlexRC can produce reduced models with favorable transient simulation efficiency. The experiments also demonstrate the benefits of the sparsity-control variant. 

\section*{Acknowledgments}
We sincerely acknowledge the data support and other technical assistance provided by Huada Empyrean Software Company Ltd. We used ChatGPT as an AI writing aid to refine the manuscript and enhance its readability. All technical content, theoretical analysis, and conclusions were developed and verified by the authors without AI assistance.

\appendices
\section{Proofs}

\subsection{Proof of \propref{prop:block coupling relation}}
\label{app:proof block coupling relation}
\begin{proof}
    From the construction of $V_2$, we have 
    \begin{align*}
        \hat{V}_2 D = \inv{\left(G_{\i} + s_2 C_{\i}\right)} B_{\i}, \qquad V_2 = \hat{V}_2\inv{T_2}. 
    \end{align*}
    Hence 
    \begin{align*}
        \left(G_{\i} + s_2 C_{\i}\right) V_2 = B_{\i} \inv{D} \inv{T_2}. 
    \end{align*}
    Thus $\spanspace{\left(G_{\i} + s_2 C_{\i}\right) V_2} \subseteq \spanspace{B_{\i}}$. The equivalent tail basis satisfies $\tran{\tilde{V}_{\mathrm{tail}}}B_{\i} = 0$. Therefore, the above inclusion gives 
    \begin{align*}
        \tran{\tilde{V}_{\mathrm{tail}}} \left(G_{\i} + s_2 C_{\i}\right) V_2 = 0, 
    \end{align*}
    or equivalently 
    \begin{align*}
        \tran{\tilde{V}_{\mathrm{tail}}}G_{\i}V_2 + s_2\tran{\tilde{V}_{\mathrm{tail}}}C_{\i}V_2 = 0. 
    \end{align*}
    By the block definitions, these two terms are the lower-left coupling blocks $\tilde{G}_{\tail,2}$ and $\tilde{C}_{\tail,2}$. Hence $\tilde{G}_{\tail,2} = -s_2\tilde{C}_{\tail,2}$. 
\end{proof}

\subsection{Proof of \thmref{thm:internal perturbation bound}}
\label{app:proof internal perturbation bound}
\begin{proof}
    Let $A_{\omega} \define G_{\i} + \imaginary \omega C_{\i}$. Since 
    \begin{align*}
        \perturbed{H}_{\i}(\imaginary \omega) - H_{\i}(\imaginary \omega) = \tran{\Delta}\inv{A_{\omega}}\tilde{C}_{\c}
        + \tran{\tilde{C}_{\c}}\inv{A_{\omega}}\Delta + \tran{\Delta}\inv{A_{\omega}}\Delta, 
    \end{align*}
    we have 
    \begin{align*}
        & \Frobenius{\perturbed{H}_{\i}(\imaginary \omega) - H_{\i}(\imaginary \omega)} \\ 
        & \quad \leq \twonorm{\inv{A_{\omega}}} \Frobenius{\Delta}
        \left( 2 \Frobenius{\tilde{C}_{\c}} + \Frobenius{\Delta} \right). 
    \end{align*}
    Moreover, for any nonzero $x \in \complex^{N-p}$, 
    \begin{align*}
        \twonorm{A_{\omega}x} \twonorm{x}
        & \geq \left| \conjtran{x}A_{\omega}x \right| 
        \geq \Re\left( \conjtran{x}A_{\omega}x \right) \\ 
        & = \conjtran{x}G_{\i}x \geq \lambda_{\min}(G_{\i}) \twonorm{x}^2, 
    \end{align*}
    Hence $\twonorm{A_{\omega}x} \geq \lambda_{\min}(G_{\i})\twonorm{x}$ for any nonzero $x$. This gives 
    \begin{align*}
        \twonorm{\inv{A_{\omega}}} \leq \dfrac{1}{\lambda_{\min}(G_{\i})}. 
    \end{align*}
    Therefore, 
    \begin{align*}
        & \Frobenius{\perturbed{H}_{\i}(\imaginary \omega) - H_{\i}(\imaginary \omega)} \\ 
        \leq\ & \dfrac{1}{\lambda_{\min}(G_{\i})} \Frobenius{\Delta} \left( 2 \Frobenius{\tilde{C}_{\c}} + \Frobenius{\Delta} \right). 
    \end{align*}

    Let $c_j$ be the $j$th column of $\tilde{C}_{\c}$. Since $c_j$ is real and 
    \begin{align*}
        \Re\left( \inv{A_{\omega}} \right)
        & = \dfrac{\inv{A_{\omega}} + \invconjtran{A_{\omega}}}{2} \\
        & = \invconjtran{A_{\omega}} G_{\i} \inv{A_{\omega}}, 
    \end{align*}
    we obtain 
    \begin{align*}
        \left| \tran{c_j}\inv{A_{\omega}}c_j \right|
        & \geq \Re\left( \tran{c_j}\inv{A_{\omega}}c_j \right) \\ 
        & = \conjtran{\left( \inv{A_{\omega}}c_j \right)}G_{\i}
        \left( \inv{A_{\omega}}c_j \right) \\ 
        & \geq \lambda_{\min}(G_{\i}) \twonorm{\inv{A_{\omega}}c_j}^2 \\ 
        & \geq \dfrac{\lambda_{\min}(G_{\i})}{\twonorm{A_{\omega}}^2} \twonorm{c_j}^2. 
    \end{align*}
    The last inequality follows from $\twonorm{c_j} = \twonorm{A_{\omega}\inv{A_{\omega}}c_j} \leq \twonorm{A_{\omega}}\twonorm{\inv{A_{\omega}}c_j}$. Therefore, 
    \begin{align*}
        \Frobenius{H_{\i}(\imaginary \omega)}
        & \geq \sqrt{\sum_{j = 1}^p \left| \tran{c_j}\inv{A_{\omega}}c_j \right|^2} \\ 
        & \geq \dfrac{\lambda_{\min}(G_{\i})}{\twonorm{A_{\omega}}^2}
        \sqrt{\sum_{j = 1}^p \twonorm{c_j}^4} \\ 
        & \geq \dfrac{\lambda_{\min}(G_{\i})}{\sqrt{p} \twonorm{A_{\omega}}^2}
        \Frobenius{\tilde{C}_{\c}}^2. 
    \end{align*}
    Hence, 
    \begin{align*}
        \delta_{\i}(\imaginary \omega)
        & \leq \sqrt{p} \dfrac{\twonorm{A_{\omega}}^2}{\lambda_{\min}(G_{\i})^2}
        \dfrac{\Frobenius{\Delta}}{\Frobenius{\tilde{C}_{\c}}} \left( 2 + \dfrac{\Frobenius{\Delta}}{\Frobenius{\tilde{C}_{\c}}} \right) \\
        & \leq \sqrt{p} \dfrac{\twonorm{A_{\omega}}^2}{\lambda_{\min}(G_{\i})^2}
        \epsilon ( 2 + \epsilon ). 
    \end{align*}
    This completes the proof. 
\end{proof}

\bibliographystyle{IEEEtran}
\bibliography{ref}

\end{document}